\documentclass[pdftex,a4paper,fleqn]{cas-sc}
\usepackage{amssymb, amsmath, amsthm, bm,amsfonts}
\usepackage[square, comma, sort&compress, numbers]{natbib}
\usepackage{graphicx}
\usepackage{subfigure} 
\newtheorem{theorem}{Theorem}
\newtheorem{definition}{Definition}

\newdefinition{rmk}{Remark}
\newproof{pf}{Proof}
\newtheorem{example}{Example}
\usepackage{caption}
\usepackage{xcolor}
\usepackage{algorithm}
\usepackage{algorithmicx}
\usepackage{algpseudocode}

\def\tsc#1{\csdef{#1}{\textsc{\lowercase{#1}}\xspace}}
\tsc{WGM}
\tsc{QE}
\tsc{EP}
\tsc{PMS}
\tsc{BEC}
\tsc{DE}

\begin{document}
\bibliographystyle{unsrt}
\let\WriteBookmarks\relax
\def\floatpagepagefraction{1}
\def\textpagefraction{.001}
\shorttitle{Cryptanalysis  for a family of PNDCCs}
\shortauthors{Q. Wang and S. Yu}

\title [mode = title]{Cryptanalysis  for a family of plaintext-non-delayed chaotic ciphers}                      
\tnotemark[1]

\tnotetext[1]{This document is the results of the research
   project funded by the National Natural Science Foundation of China (No. 62271157).}

   \author[1]{Qianxue Wang}[orcid=0000-0002-4063-5401]
\ead{wangqianxue@gdut.edu.cn}
\ead[url]{https://teacher.gdut.edu.cn/wangqianxue/en/index.htm}

\author[1]{Simin Yu}[orcid=0000-0002-1192-6955]
\cormark[1]
\ead{siminyu@163.com}
\address[1]{School of Automation, Guangdong University of Technology, Guangzhou 510006, China}

\cortext[cor1]{Corresponding author}

\begin{abstract}
  Plaintext-non-delayed chaotic cipher (PNDCC) is defined such that, in the diffusion equation, plaintext has no delay terms, while the ciphertext has a feedback term. Its general form can be expressed as  $C(k)={{f}_{D}}\left( P(k),C(k-1),K(k) \right)$, where ${{f}_{D}}$  is the iteration function. In the existing literature, the diffusion functions of chaotic ciphers adopt this form without exception. Since the introduction of PNDCC, its security performance has attracted wide attention but also raised serious doubts. The main reason lies in the fact that designers of chaotic ciphers generally demonstrate the ``security'' of PNDCC through the fulfillment of various statistical criteria, while rigorous security proofs in the sense of modern cryptography are absent. Consequently, there is a good reason to critically re-examine and question both the design motivation and the empirical security of PNDCC. To address this issue, we present a typical example of a three-stage permutation-diffusion-permutation PNDCC, which contains multiple security vulnerabilities. Although all of its statistical indicators show good performance, we are able to break it using four different attacks. The first is a differential attack based on homogeneous operations; the second is an S-PTC attack; the third is a novel impulse-step-based differential attack (ISBDA), proposed in this paper, and the fourth is a novel chain attack, also introduced here. These results demonstrate that the fulfillment of statistical criteria is not a sufficient condition for the security of PNDCC. Then, based on a mathematical model of multi-stage PNDCC, we show that the proposed chain attack can successfully break a class of multi-stage PNDCCs. The key technique of the chain attack depends on how to reveal all permutations. To address this key problem, we summarize the chaining rules and show that, from the attacker's perspective, if the same decryption chain can be reconstructed then all permutations can be deciphered. To that end, the entire diffusion process can be broken by solving a system of simultaneous equations. Finally, as a secure improvement, we propose a new scheme termed plaintext-delayed chaotic cipher (PDCC) that can resist various cryptanalytic attacks. Several case studies, supported by both theoretical analysis and numerical experiments, and also together with the directly runnable MATLAB program provided by this paper, confirm the reliability and effectiveness of the proposed method.
\end{abstract}

\begin{highlights}
  \item Propose a chain attack to recover all permutations of multi‑stage PNDCC.

  \item Reveal that satisfying statistical indicators is not sufficient for chaos‑based cipher security.

  \item Introduce the plaintext‑delayed chaotic cipher (PDCC) to defeat the identified attacks.
\end{highlights}

\maketitle

\section{Introduction}
In recent years, chaos has demonstrated potential application value in the field of dependable and secure computing, spanning network security, software security, internet-of-things security, cloud-computing security, and more~\cite{ref1,ref2,ref3}. With respect to multimedia security, applications of chaos mainly concern the analysis and design of chaos-based image encryption, a research direction that has attracted sustained and broad attention over a long period.

At present, the reason why a large number of chaotic ciphers have been successively broken~\cite{ref26,ref27,ref28,ref29,ref30,ref31,ref32} lies primarily in improper design, which leads to a series of evident security vulnerabilities exploitable by attackers. Fundamentally, this does not stem from any inherent defect of chaotic ciphers themselves, but rather from the fact that some designers, without a proper understanding of chaotic cryptanalysis, hastily embark on cipher design. For example, by analyzing the specific forms of the diffusion function $f_{D}$ in~\cite{ref13,ref14,ref15,ref16,ref17,ref18,ref19,ref20,ref21,ref22,ref23,ref24}, \cite{ref26} and \cite{ref27} observed that the diffusion functions in \cite{ref13,ref14,ref17,ref19} adopt heterogeneous operations, namely mixtures of modular and XOR operations, whereas those in~\cite{ref15,ref16,ref18,ref20,ref21,ref22,ref23,ref24} adopt homogeneous operations, including homogeneous modular or homogeneous XOR operations. Building on this, \cite{ref26} and \cite{ref27} were the first to propose a differential attack based on homogeneous operations to break chaotic ciphers whose diffusion functions are homogeneous. By analyzing the encryption algorithm in~\cite{ref12}, \cite{ref28} proposed employing multiple attack methods. First, \cite{ref28} found that the five encryption stages in the original algorithm reduced to a two-stage diffusion-permutation, and the resulting aggregate diffusion stage is homogeneous, enabling a homogeneous-operation-based differential attack to break the algorithm. Then, \cite{ref28} discovered that the equivalent aggregate diffusion exhibits a unidirectional diffusion property, and by using plaintext intersection to determine the positional correspondence between the plaintext set and the ciphertext set, an S-PTC attack was pioneered to successfully break the algorithm proposed in~\cite{ref12}.
In particular, even to this day, there are still a large number of reports that prove the security of the designed chaotic ciphers by meeting various statistical indicators, such as histogram, correlation, NPCR, UACI, and information entropy~\cite{ref33,ref34,ref35}.  However, our findings point out that even if these statistical indicators meet common thresholds, this does not necessarily mean that they are sufficient to reflect true security.  Indeed, algorithms with favorable histogram, NPCR, and entropy metrics have been completely broken by chosen-plaintext attacks~\cite{ref5,ref6}, while strong S-Box designs require eliminating fixed points and short period rings that statistical tests fail to detect~\cite{ref7}. Therefore, the current security evaluation methodologies for chaotic ciphers should be thoroughly rethought.

Although many chaotic ciphers have been successively broken~\cite{ref26,ref27,ref28,ref29,ref30,ref31,ref32}, revealing a variety of vulnerabilities such as homogeneous operations and one-way sensitivity and achieving several important advances, all the analytical results in~\cite{ref26,ref27,ref28,ref29,ref30,ref31,ref32} merely highlight individual cases of security hazards and flaws in current chaotic-cipher designs. These reflect only one aspect of the particularity of the security problem, without essentially uncovering a common issue that pervasively exists across current chaotic ciphers. This belongs to another aspect of the generally existing security problems of chaotic ciphers and calls for a fundamental solution: namely, why is it that for all current chaotic ciphers, despite various apparent security vulnerabilities, they can still readily pass all statistical tests such as histogram, correlation, NPCR, UACI, and information entropy~\cite{ref5,ref6,ref7,ref10,ref11,ref12,ref13,ref14,ref15,ref16,ref17,ref18,ref19,ref20,ref21,ref22,ref23,ref24,ref33,ref34,ref35}?

This paper conducts an in-depth study of the above issue and, for the first time, observes that the diffusion equations of all chaotic ciphers proposed in the existing literature contain no plaintext delay term $P(k-1)$. We refer to such schemes as plaintext-non-delayed chaotic ciphers (PNDCCs)~\cite{ref8,ref9,ref10,ref11,ref12,ref13,ref14,ref15,ref16,ref17,ref18,ref19,ref20,ref21,ref22,ref23,ref24}, whose diffusion equation generally takes the form
$C(k)={{f}_{D}}\left( P(k),C(k-1),K(k) \right)$,
where $f_{D}$ is the diffusion function. First, numerical simulations show that in all PNDCCs, the common essential characteristic is the absence of a plaintext delay term in the diffusion equation, while a ciphertext feedback term is present. The historical influence of the plaintext is indirectly introduced into the current diffusion through ciphertext feedback; the ciphertext or internal state depends not only on the current plaintext, but also—via previously generated ciphertext or state—indirectly on the entire plaintext history. In other words, even if the diffusion equation contains no explicit plaintext delay term, the influence of all historical plaintext is still implicitly preserved by ciphertext feedback. Therefore, even in the presence of many other security vulnerabilities, such schemes can still pass all statistical tests. In short, as long as there is a ciphertext feedback term, the statistical indicators will definitely meet the standards. To illustrate this point, we first present a three-stage permutation-diffusion-permutation PNDCC with multiple security vulnerabilities, and we break it using four attack methods, including ISBDA and the chain attack. This demonstrates that “meeting” statistical indicators is only a necessary condition, not a sufficient one, for the security of PNDCCs. Next, we establish a general model for multi-stage PNDCCs and derive the associated generalized iterative equation, showing that, from an attacker’s perspective, a class of PNDCCs can be successfully broken in one stroke via the proposed chain attack. The main feature of the chain attack is that it first recovers all permutation keys and then all diffusion keys; it is simple to carry out and has low attack complexity, offering advantages unmatched by conventional cryptanalytic methods. The core technique of the chain attack lies in how to first recover all positional permutations. To address this key problem, we summarize the chaining rules and, from the attacker’s perspective, reconstruct the same decryption chain formed during decryption, and then use it to decipher all permutations, thereby making it possible to first recover all positional permutations. We note that the chain attack and the S-PTC attack in \cite{ref28} share the common requirement of establishing positional correspondence between plaintext and ciphertext sets via a chosen-ciphertext attack. The difference is that in the S-PTC attack, once positional correspondence is found, one must write all possible values and build a codebook between ciphertext and plaintext. With such a codebook, an attacker can obtain the corresponding plaintext for any ciphertext by lookup. Consequently, compared with the chain attack, the S-PTC attack entails much higher complexity. For the chain attack, one only needs to solve the one-dimensional problem of pixel positions first, and then break the entire diffusion by solving a system of simultaneous equations; in contrast, the S-PTC attack must handle both pixel magnitude and pixel position—two dimensions. Finally, to fundamentally address the security vulnerabilities of current chaotic ciphers, we further propose, as an improved scheme, a plaintext-delayed chaotic cipher (PDCC) that can withstand various attacks on cryptographic algorithms.

The remainder of this paper is organized as follows. Sec.~\ref{sec:3level-pndcc-stat-isbda} introduces the statistical  indicators and the ISBDA cryptanalysis of the three-stage permutation-diffusion-permutation PNDCC. Sec.~\ref{sec:general-iter-eq} presents the generalized iterative equation and the associated position expression for PNDCCs. Sec.~\ref{sec:chained-attack} describes the chain attack against PNDCCs and its implementation. Sec.~\ref{sec:diffusion-attack} introduces the cryptanalysis of the diffusion stages after all permutation stages have been recovered. Sec.~\ref{sec:attack-example-pdcc} provides a worked example of the chain attack and an improved scheme, PDCC, for PNDCCs. Sec.~\ref{sec:conclusion} concludes the paper.

\section{Statistical Indicators of Three-Stage Permutation-Diffusion-Permutation PNDCC and  its ISBDA}
\label{sec:3level-pndcc-stat-isbda}
In the extensive literature on chaotic cipher design~\cite{ref10,ref11,ref12,ref13,ref14,ref15,ref16,ref17,ref18,ref19,ref20,ref21,ref22,ref23,ref24,ref33,ref34,ref35}, designers commonly substantiate the “security” of their schemes by meeting a set of statistical indicators. However, in this section, through the security analysis of a representative three-stage permutation-diffusion-permutation PNDCC, we show that even when these statistical indicators are satisfied, one cannot conclude that the cipher is secure. Meeting statistical indicators is only one of the necessary conditions for the security of chaotic ciphers; it is not sufficient.  This finding likewise applies to modern cryptography.

\subsection{Description of the Three-Stage Permutation-Diffusion-Permutation PNDCC}
A three-stage permutation-diffusion-permutation PNDCC is shown in Fig.~\ref{fig1}, where the diffusion equation is
\begin{equation}
    \label{eq1}
    C(k)=\bmod (P(k)-C(k-1)-K(k),256).
\end{equation}

\begin{figure}[htbp]
  \centering 
  \includegraphics[width=0.4\textwidth]{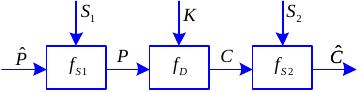} 
  \caption{Block diagram of a three-stage permutation-diffusion-permutation PNDCC with diffusion satisfying Eq.~\eqref{eq1}.} 
  \label{fig1} 
\end{figure}
Here, $k=1,2,\ldots,M\times N$, and $M\times N$ is the image size. It is evident that Eq.~\eqref{eq1} satisfies the general form of the diffusion equation $C(k)=f_{D}\left(P(k),C(k-1),K(k)\right)$. Based on the scheme in Fig.~\ref{fig1}, after two rounds of computation, the corresponding statistical analysis indicators are evaluated, including: (1) histograms of the plaintext and ciphertext; (2) horizontal, vertical, and diagonal correlations of the plaintext and ciphertext; (3) NPCR and UACI; and (4) the entropy of the plaintext and ciphertext.

\subsection{Computation of Statistical Indicators}
(1) Histogram analysis: The histograms of the plaintext and ciphertext are shown in Figs.~\ref{fig2}.  

\begin{figure}[!ht]
  \centering
  \begin{minipage}[b]{0.48\linewidth}
    \centering
    \includegraphics[width=\linewidth]{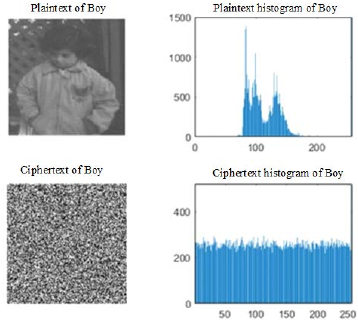}
    \vspace{0.3em}
    { (a) Histograms of the Boy plaintext/ciphertext.}
  \end{minipage}
  \hfill
  \begin{minipage}[b]{0.48\linewidth}
    \centering
    \includegraphics[width=\linewidth]{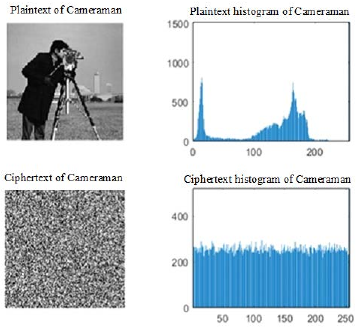}
    \vspace{0.3em}
    { (b) Histograms of the Cameraman plaintext/ciphertext.}
  \end{minipage}
  \caption{Histograms of plaintext/ciphertext for different images.}
  \label{fig2}
\end{figure}


(2) Correlation analysis: The correlation analysis results for the plaintext and ciphertext are summarized in Tab.~\ref{tab1}.

\begin{table}[!h]
  \centering
  \caption{Correlation analysis}
  \label{tab1}
  \begin{tabular}{lllll}
    \hline
    Image                       & Plain/Cipher & Horizontal  &  Vertical & Diagonal    \\ \hline
    \multirow{2}{*}{Boy}       & Plaintext    & 0.9007  & 0.9002 & 0.8301  \\ \cline{2-5} 
                               & Ciphertext    & -0.0131 & 0.0161 & -0.0146 \\ \hline
    \multirow{2}{*}{Cameraman} & Plaintext    & 0.9090  & 0.9344 & 0.8787  \\ \cline{2-5} 
                               & Ciphertext    & 0.0325  & 0.0234 & 0.0344  \\ \hline
    \end{tabular}
\end{table}

(3) NPCR and UACI analysis: The NPCR and UACI results for the plaintext and ciphertext are given in Tab.~\ref{tab2}.

\begin{table}[htbp]
  \centering
  \caption{NPCR and UACI analysis}
  \label{tab2}
  \begin{tabular}{llll}
    \hline
    Image & Influence between plain and cipher & NPCR & UACI      \\ \hline
    Boy & Plaintext affects ciphertext     & 99.1782\% & 33.4051\%  \\ \hline
    Cameraman & Plaintext affects ciphertext     & 99.4213\% & 33.5664\% \\ \hline
    \end{tabular}
\end{table}

(4) Information entropy analysis: The information entropy of the plaintext and ciphertext is reported in Tab.~\ref{tab3}.
\begin{table}[htbp]
  \centering
  \caption{Information entropy analysis}
  \label{tab3}
  \begin{tabular}{llll}
    \hline
    \multicolumn{2}{C}{Boy}             & \multicolumn{2}{C}{Cameraman}       \\ \hline
    \multicolumn{1}{l}{Plaintext}     & Ciphertext     & \multicolumn{1}{l}{Plaintext}     & Ciphertext      \\ \hline
    \multicolumn{1}{l}{3.0622} & 7.9866 & \multicolumn{1}{l}{4.9295} & 7.9874 \\ \hline
    \end{tabular}
\end{table}

\subsection{ISBDA Cryptanalysis of the Three-Stage PNDCC}
From Fig.~\ref{fig2} and Tabs.~\ref{tab1}-\ref{tab3}, it is evident that the three-stage permutation-diffusion-permutation PNDCC in Fig.~\ref{fig1} meets all statistical indicators after two rounds. However, according to Eq.~\eqref{eq1}, although the diffusion equation is sensitive from ciphertext to plaintext, it has two vulnerabilities: homogeneous operations and insensitivity from plaintext to ciphertext. 
The homogeneous-operations vulnerability can be exploited by a differential attack based on homogeneous operations~\cite{ref26,ref27,ref29,ref30}. The plaintext-to-ciphertext insensitivity can be exploited by two methods: the first is an S-PTC attack~\cite{ref28}, and the second is a chain attack to be introduced in Sec.~\ref{sec:chained-attack}. Moreover, noting that Fig.~\ref{fig1} adopts a three-stage permutation-diffusion-permutation structure, the ISBDA proposed in this section can also be applied.

The main feature of ISBDA is that, for any three-stage permutation-diffusion-permutation PNDCC, as long as the input-output characteristic of its diffusion equation satisfies the impulse-input/step-output condition, ISBDA can be used to break it. In short, ISBDA provides a new avenue for breaking this class of chaotic ciphers.

(1) Basic principle

For the three-stage permutation-diffusion-permutation PNDCC in Fig.~\ref{fig1}, the corresponding diffusion equation satisfies the impulse-input/step-output condition from $P$ to $C$, as shown in Fig.~\ref{fig4}. Therefore, an impulse-step-based differential attack can be used to break the three-stage permutation-diffusion-permutation PNDCC.
\begin{figure*}[htbp]
  \centering
  \includegraphics[width=0.85\textwidth]{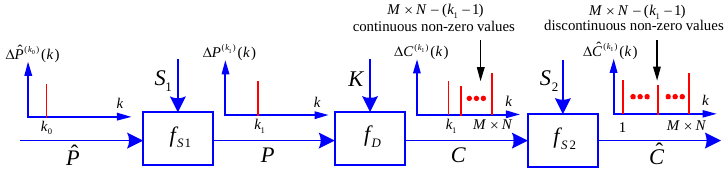}
  \caption{Block diagram of breaking a three-stage permutation-diffusion-permutation PNDCC using a left-to-right impulse-step-based differential attack.}
  \label{fig4}
  \end{figure*}

According to Fig.~\ref{fig4}, let the differential input be an impulse function, whose mathematical expression is
\begin{equation}
  \label{eq2}
  \left\{ \begin{aligned}
    & {{{\hat{P}}}^{({{V}_{1}})}}(k)={{V}_{1}} \\ 
   & {{{\hat{P}}}^{({{k}_{0}})}}(k)=\left\{ \begin{aligned}
    & {{V}_{2}}\ne {{V}_{1}}\text{  (if }k={{k}_{0}}) \\ 
   & {{V}_{2}}={{V}_{1}}\text{  (if }k\ne {{k}_{0}}) \\ 
  \end{aligned} \right. \\ 
   & \Delta {{{\hat{P}}}^{({{k}_{0}})}}(k)={{{\hat{P}}}^{({{k}_{0}})}}(k)\otimes {{{\hat{P}}}^{({{V}_{1}})}}(k)=\left\{ \begin{aligned}
    & {{V}_{2}}\otimes {{V}_{1}}\text{  (if }k={{k}_{0}}) \\ 
   & 0\text{          (if }k\ne {{k}_{0}}) \\ 
  \end{aligned} \right. \\ 
  \end{aligned} \right.,
\end{equation}
where $V_{1},V_{2}\in\{0,1,2,\ldots,255\}$, $\otimes \in \{\oplus ,\dot{-}\}$, and $k,k_{0}=1,2,\ldots,M\times N$. After the pre-permutation, the impulse becomes
\begin{equation}
  \label{eq3}
  \begin{aligned}
    \Delta {{P}^{({{k}_{1}})}}(k)&={{P}^{({{k}_{1}})}}(k)\otimes {{P}^{({{V}_{1}})}}(k)\\
  &=\left\{ \begin{aligned}
  & {{V}_{2}}\otimes {{V}_{1}}\text{  (if }k={{k}_{1}}) \\ 
 & 0\text{          (if }k\ne {{k}_{1}}) \\ 
\end{aligned} \right.,
\end{aligned}
\end{equation}    
where $\otimes \in \{\oplus ,\dot{-}\}$, $V_{1},V_{2}\in\{0,1,2,\cdots ,255\}$, $k=1,2,\ldots,M\times N$, and $k_{1}\in\{1,2,\cdots ,M\times N\}$. As $k_{0}$ traverses all positions $k_{0}\in\{1,2,\cdots ,M\times N\}$, $k_{1}$ correspondingly traverses all positions $k_{1}\in\{1,2,\cdots ,M\times N\}$. In this case, if the differential output of the diffusion stage satisfies                                 
\begin{equation}
  \label{eq4}
  \begin{aligned}
  \Delta {{C}^{({{k}_{1}})}}(k)&={{C}^{({{k}_{1}})}}(k)\otimes {{C}^{({{V}_{1}})}}(k)\\
  &=\left\{ \begin{aligned}
  & 0\text{        (if }k<{{k}_{1}}) \\ 
 & \text{not }0\text{  (if }k\ge {{k}_{1}}) \\ 
\end{aligned} \right.,
\end{aligned}
\end{equation}                                       
then note that the post-permutation does not change the magnitudes of $\Delta C^{(k_{1})}(k)$; it only changes the positions of nonzero pixels. Therefore, by simply counting the number of nonzero pixels at the output, one can determine the value of $k_{1}$ and thus locate its position. Since the position $k_{0}$ is known, comparing the positional change between $k_{0}$ and $k_{1}$ allows the pre-permutation to be recovered.

Once the pre-permutation is recovered, the intermediate variable $P$ becomes known. On this basis, choose an all-zero ciphertext $\widehat{C}=0$ to neutralize the post-permutation, yielding $C=0$. Under a chosen-ciphertext attack, the plaintext $\widehat{P}$ corresponding to $\widehat{C}=0$ is known; furthermore, since the pre-permutation has been recovered, $P$ is also known. Evidently, when both $P$ and $C$ are known, the intermediate diffusion stage can be broken. With both the pre-permutation and the intermediate diffusion recovered, $C$ becomes known as well, and the post-permutation can then be deciphered. By thus defeating each stage in turn, the entire three-stage permutation-diffusion-permutation PNDCC in Fig.~\ref{fig1} is ultimately broken.

(2) Proof that the diffusion stage in Fig.~\ref{fig1} satisfies the impulse-input/step-output condition

From Fig.~\ref{fig1}, for $k=1,2,\ldots,M\times N$, the expressions of $P^{(V_{1})}(k)$, $P^{(k_{1})}(k)$, and $\Delta P^{(k_{1})}(k)$ are
\begin{equation}
  \label{eq5}
\left\{ \begin{aligned}
  & {{P}^{({{V}_{1}})}}(k)={{V}_{1}} ,\\
  &{{P}^{({{k}_{1}})}}(k)=\left\{ \begin{aligned}
  & {{V}_{2}}\ne {{V}_{1}}\text{  (if }k={{k}_{1}}) \\ 
 & {{V}_{2}}={{V}_{1}}\text{  (if }k\ne {{k}_{1}}) \\ 
\end{aligned} \right. \\ 
 & \begin{aligned}\Delta {{P}^{({{k}_{1}})}}(k)&={{P}^{({{k}_{1}})}}(k)\dot{-}{{P}^{({{V}_{1}})}}(k)\\
  &=\left\{ \begin{aligned}
  & {{V}_{2}}\dot{-}{{V}_{1}}\text{  (if }k={{k}_{1}}) \\ 
 & 0\text{          (if }k\ne {{k}_{1}}) \\ 
\end{aligned} \right. \\ 
\end{aligned}
\end{aligned} \right.,
\end{equation} 
where ${{V}_{1}}\in \{0,1,2,3,\cdots ,255\}$, $k=1,2,\ldots,M\times N$, $k_{1}\in\{1,2,\ldots,M\times N\}$, $V_{2}\in\{0,1,2,\ldots,255\}$ with $V_{2}\ne V_{1}$, $P^{(V_{1})}(0)=P^{(k_{1})}(0)$, $C^{(V_{1})}(0)=C^{(k_{1})}(0)$, and $K(0)$ is the initial value. From Eq.~\eqref{eq1} and Eq.~\eqref{eq5}, the corresponding expressions of $C^{(V_{1})}(k)$, $C^{(k_{1})}(k)$, and $\Delta C^{(k_{1})}(k)$ are
\begin{equation}
  \label{eq6}
\begin{aligned}
  & \left\{ \begin{aligned}
  & {{C}^{({{k}_{1}})}}(k)=\bmod ({{P}^{({{k}_{1}})}}(k)-{{C}^{({{k}_{1}})}}(k-1)-K(k),256) \\ 
 & {{C}^{({{V}_{1}})}}(k)=\bmod ({{P}^{({{V}_{1}})}}(k)-{{C}^{({{V}_{1}})}}(k-1)-K(k),256) \\ 
\end{aligned} \right. \\ 
 & \to \Delta {{C}^{({{k}_{1}})}}(k)={{C}^{({{k}_{1}})}}(k)\dot{-}{{C}^{({{V}_{1}})}}(k)\\
 & =\Delta {{P}^{({{k}_{1}})}}(k)\dot{-}\Delta {{C}^{({{k}_{1}})}}(k-1) \\ 
 & \to \left\{ \begin{aligned}
  & \Delta {{C}^{({{k}_{1}})}}(1)={{C}^{({{k}_{1}})}}(1)\dot{-}{{C}^{({{V}_{1}})}}(1)=\Delta {{P}^{({{k}_{1}})}}(1)\dot{-}\Delta {{C}^{({{k}_{1}})}}(0)=0 \\ 
 & \Delta {{C}^{({{k}_{1}})}}(2)={{C}^{({{k}_{1}})}}(2)\dot{-}{{C}^{({{V}_{1}})}}(2)=\Delta {{P}^{({{k}_{1}})}}(2)\dot{-}\Delta {{C}^{({{k}_{1}})}}(1)=0 \\ 
 & \cdots  \\ 
 & \Delta {{C}^{({{k}_{1}})}}({{k}_{1}}-1)={{C}^{({{k}_{1}})}}({{k}_{1}}-1)\dot{-}{{C}^{({{V}_{1}})}}({{k}_{1}}-1)=\Delta {{P}^{({{k}_{1}})}}({{k}_{1}}-1)\dot{-}\Delta {{C}^{({{k}_{1}})}}({{k}_{1}}-2)=0 \\ 
 & \Delta {{C}^{({{k}_{1}})}}({{k}_{1}})={{C}^{({{k}_{1}})}}({{k}_{1}})\dot{-}{{C}^{({{V}_{1}})}}({{k}_{1}})=\Delta {{P}^{({{k}_{1}})}}({{k}_{1}})\dot{-}\Delta {{C}^{({{k}_{1}})}}({{k}_{1}}-1) \\ 
  & \phantom{\Delta {{C}^{({{k}_{1}})}}({{k}_{1}})}={{V}_{2}}\dot{-}{{V}_{1}}\ne 0 \\ 
 & \Delta {{C}^{({{k}_{1}})}}({{k}_{1}}+1)={{C}^{({{k}_{1}})}}({{k}_{1}}+1)\dot{-}{{C}^{({{V}_{1}})}}({{k}_{1}}+1)=\Delta {{P}^{({{k}_{1}})}}({{k}_{1}}+1)\dot{-}\Delta {{C}^{({{k}_{1}})}}({{k}_{1}}) \\ 
 & \phantom{\Delta {{C}^{({{k}_{1}})}}({{k}_{1}}+1)}=256-({{V}_{2}}\dot{-}{{V}_{1}})\ne 0 \\ 
 & \Delta {{C}^{({{k}_{1}})}}({{k}_{1}}+2)={{C}^{({{k}_{1}})}}({{k}_{1}}+2)\dot{-}{{C}^{({{V}_{1}})}}({{k}_{1}}+2)=\Delta {{P}^{({{k}_{1}})}}({{k}_{1}}+2)\dot{-}\Delta {{C}^{({{k}_{1}})}}({{k}_{1}}+1) \\ 
 & \phantom{\Delta {{C}^{({{k}_{1}})}}({{k}_{1}}+2)}=({{V}_{2}}\dot{-}{{V}_{1}})\ne 0 \\ 
 & \cdots  \\ 
 & \Delta {{C}^{({{k}_{1}})}}(M\times N)={{C}^{({{k}_{1}})}}(M\times N)\dot{-}{{C}^{({{V}_{1}})}}(M\times N)=\Delta {{P}^{({{k}_{1}})}}(M\times N)\dot{-}\Delta {{C}^{({{k}_{1}})}}(M\times N-1) \\ 
 & \phantom{\Delta {{C}^{({{k}_{1}})}}(M\times N)}=({{V}_{2}}\dot{-}{{V}_{1}})\text{  or  (}256-({{V}_{2}}\dot{-}{{V}_{1}}))\ne 0  
\end{aligned} \right.. 
\end{aligned}
\end{equation} 
Based on Eq.~\eqref{eq5}-\eqref{eq6}, it suffices to choose $V_{2}\ne V_{1}$ (with $V_{2},V_{1}\in\{0,1,2,\ldots,255\}$). It is then clear that Eq.~\eqref{eq5} and Eq.~\eqref{eq6} satisfy the impulse-input/step-output condition, and thus the entire three-stage permutation-diffusion-permutation PNDCC shown in Fig.~\ref{fig1} can be broken.

\subsection{Algorithm Flowchart of ISBDA in the General Case}
In general, the diffusion equation in Fig.~\ref{fig1} may be fairly complicated, and determining whether it satisfies the impulse-input/step-output condition by the method given in Eq.~\eqref{eq6} is often difficult. To address this, we present the ISBDA algorithm flowchart for the general case, as shown in Fig.~\ref{fig5}.
\begin{figure}[htbp] 
  \centering 
  \includegraphics[width=0.58\textwidth]{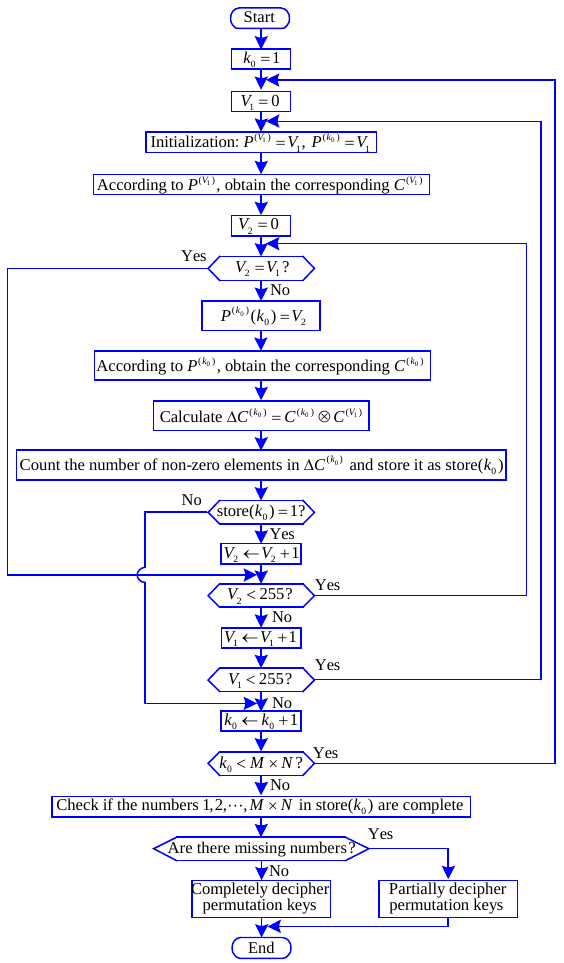} 
  \caption{Flowchart of the ISBDA algorithm} 
  \label{fig5} 
\end{figure}

\subsection{Analysis of the Impulse-Input/Step-Output Condition}
In the diffusion equation of the three-stage permutation-diffusion-permutation PNDCC 
(Fig.~\ref{fig1}), the plaintext term has no delay, while the ciphertext has a feedback term, generally expressed as $C(k)=f_{D}\big(P(k),C(k-1),K(k)\big)$. Expanding $C(k)=f_{D}\big(P(k),C(k-1),K(k)\big)$ for $k=1,2,\ldots,M\times N$ yields
\begin{equation}
  \label{eq7}
  \left\{ \begin{aligned}
  & C(1)={{f}_{D}}\left( P(1),C(0),K(1) \right)={{g}_{1}}\left( P(1) \right) \\ 
 & C(2)={{f}_{D}}\left( P(2),C(1),K(2) \right)={{g}_{2}}\left( P(1),P(2) \right) \\ 
 & \cdots  \\ 
 & C(k)={{f}_{D}}\left( P(k),C(k-1),K(k) \right)={{g}_{k}}\left( P(1),P(2),P(3),\cdots ,P(k) \right) \\ 
 & \cdots  \\ 
 & C(M\times N)={{f}_{D}}\left( P(M\times N),C(M\times N-1),K(M\times N) \right)={{g}_{M\times N}}\left( P(1),P(2),P(3),\cdots ,P(M\times N) \right) \\ 
\end{aligned} \right..
\end{equation}
Note that the functions $g_{i}$ ($i=1,2,\ldots,M\times N$) describe only the relationship between plaintext and ciphertext under key-fixed conditions. According to cryptanalytic convention, the key and initial conditions are treated as invariants throughout the analysis, and thus are not included explicitly in $g_{i}$. The above shows that when the plaintext $P(k)$ ($k=1,2,\ldots,M\times N$) has a perturbation at position $k$, and if the diffusion function $f_{D}$ ensures
\begin{equation}
  \label{eq8}
\left\{ \begin{aligned}
  & C(1)={{f}_{D}}\left( P(1),C(0),K(1) \right)={{g}_{1}}\left( P(1) \right)=0 \\ 
 & C(2)={{f}_{D}}\left( P(2),C(1),K(2) \right)={{g}_{2}}\left( P(1),P(2) \right)=0 \\ 
 & \cdots  \\ 
 & C(k-1)={{f}_{D}}\left( P(k-1),C(k-2),K(k-1) \right)={{g}_{k-1}}\left( P(1),P(2),P(3),\cdots ,P(k-1) \right)=0 \\ 
 & C(k)={{f}_{D}}\left( P(k),C(k-1),K(k) \right)={{g}_{k}}\left( P(1),P(2),P(3),\cdots ,P(k) \right)\ne 0 \\ 
 & C(k+1)={{f}_{D}}\left( P(k+1),C(k),K(k+1) \right)={{g}_{k+1}}\left( P(1),P(2),P(3),\cdots ,P(k),P(k+1) \right)\ne 0 \\ 
 & \cdots  \\ 
 & C(M\times N)={{f}_{D}}\left( P(M\times N),C(M\times N-1),K(M\times N) \right)={{g}_{M\times N}}\left( P(1),P(2),P(3),\cdots ,P(M\times N) \right)\ne 0 \\ 
\end{aligned} \right.,
\end{equation}
then the ciphertext samples $C(k),C(k+1),\ldots,C(M\times N)$ at positions $k,k+1,\ldots,M\times N$ contain exactly $M\times N+1-k$ nonzero perturbation outputs. This guarantees that when the plaintext has an impulse input at position $k$ ($k=1,2,\ldots,M\times N$), the corresponding ciphertext exhibits a step output over positions $k,k+1,\ldots,M\times N$ ($k=1,2,\ldots,M\times N$).

\subsection{Simulation Experiments for Breaking the Three-Stage Permutation-Diffusion-Permutation PNDCC}
\begin{example}
Following the ISBDA algorithm flow  shown in Fig.~\ref{fig5}, we break the three-stage permutation-diffusion-permutation PNDCC whose diffusion equation is given by Eq.~\eqref{eq1}. 
\end{example}
\begin{example}
    Following the ISBDA algorithm flow  shown in Fig.~\ref{fig5}, we break the three-stage permutation-diffusion-permutation PNDCC whose diffusion equation is a heterogeneous composite-operation diffusion equation, as
    \begin{equation}
      \label{eq8-1}
      C(k)= \bmod(P(k)+K(k),256)\oplus C(k-1)\oplus K(k-1).
    \end{equation}     
\end{example}
The simulation results are shown in Fig.\ref{fig6}. Both examples are successfully decrypted. Since the original plaintext image is identical in the two simulations, the recovered images are the same; we therefore present two distinct encrypted images. For ISBDA, the complexity is $M \times N + 1$ plaintext–ciphertext pairs. It should be emphasized that the proposed algorithm is also applicable to diffusion equations in the general case.
\begin{figure}[htbp]
  \centering
\begin{minipage}{0.2\linewidth}
\centering
\includegraphics[width=\linewidth]{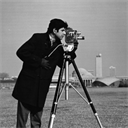}\\
a)
\end{minipage}
\begin{minipage}{0.2\linewidth}
\centering
\includegraphics[width=\linewidth]{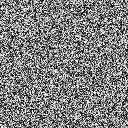}\\
b)
\end{minipage}
\begin{minipage}{0.2\linewidth}
    \centering
    \includegraphics[width=\linewidth]{figs/encrypted128.png}
    c)
    \end{minipage}
\begin{minipage}{0.2\linewidth}
\centering
\includegraphics[width=\linewidth]{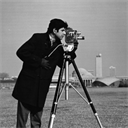}
d)
\end{minipage}
  \caption{Breaking results for Cameraman: (a) plaintext; (b) encrypted by Eq.~\eqref{eq1}; (c) encrypted by Eq.~\eqref{eq9}; (d) deciphered.}
  \label{fig6}
\end{figure}

\section{Generalized Iterative Equation and Associated Position Expression for Multi-Stage PNDCC}
\label{sec:general-iter-eq}
In Sec.~\ref{sec:3level-pndcc-stat-isbda}, we analyzed the security of the three-stage permutation-diffusion-permutation PNDCC. In this section, we first derive generalized iterative equations and  associated position expressions for multi-stage PNDCCs, thereby providing a universal mathematical model to support the security analysis in Sec.~\ref{sec:chained-attack}. For brevity, we take as examples the three-stage diffusion–permutation–diffusion PNDCC, the four-stage permutation–diffusion–permutation–diffusion PNDCC, the four-stage diffusion–permutation–diffusion–permutation PNDCC, and the five-stage permutation–\allowbreak diffusion–\allowbreak permutation–diffusion–\allowbreak permutation PNDCC, and derive their generalized iterative equations and associated position expressions; the method readily extends to arbitrary multi-stage PNDCCs without loss of generality. Moreover, since encryption and decryption are inverses, we present the generalized decryption iterative equation; the corresponding generalized encryption equation follows by inversion. Unless otherwise stated, “generalized iterative equation” refers to the generalized decryption iterative equation.

\subsection{Generalized Iterative Equation for Multi-Stage PNDCC}
Let the decryption diagram for multi-stage PNDCC, including  the three-stage diffusion–permutation–diffusion PNDCC, the four-stage permutation–diffusion–permutation–diffusion PNDCC, the four-stage diffusion\allowbreak–permutation\allowbreak–diffusion\allowbreak–permutation PNDCC, and the five-stage permutation-diffusion-permutation-diffusion-permutation structure, be as shown in Fig.~\ref{fig7}. 

The general form of the pre-diffusion inverse in Fig.~\ref{fig7} is
\begin{equation}
  \label{eq9}
{{P}_{1}}\left( k \right)=f_{D1}^{-1}\left({{C}_{1}}\left( k \right),{{C}_{1}}\left( k-1 \right),{{C}_{1}}\left( k-2 \right),\cdots ,{{C}_{1}}\left( k-{{m}_{1}}\right),{{K}_{1}}(k) \right),
\end{equation}
and the general form of the post-diffusion inverse is
\begin{equation}
  \label{eq10}
{{P}_{2}}\left( k \right)=f_{D2}^{-1}\left(
{{C}_{2}}\left( k \right),{{C}_{2}}\left( k-1 \right),{{C}_{2}}\left( k-2 \right),\cdots , {{C}_{2}}\left( k-{{m}_{2}} \right),{{K}_{2}}(k) \right),
\end{equation}
where $m_{1},m_{2}$ are positive integers and $k=1,2,\ldots,M\times N$.
\begin{figure*}[htbp] 
  \centering 
  \includegraphics[width=0.85\textwidth]{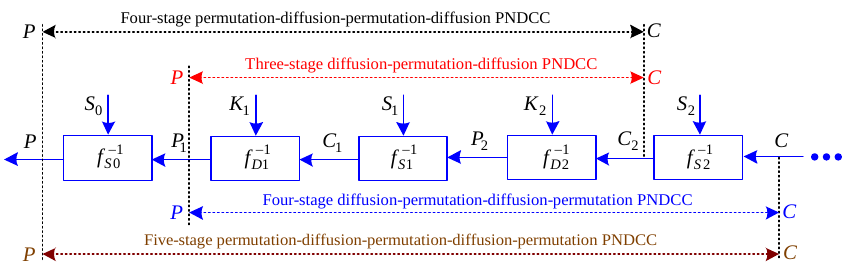}
  \caption{Decryption block diagram for multi-stage PNDCC.} 
  \label{fig7}
\end{figure*}

(1) From  Fig.~\ref{fig7} and Eqs.~\eqref{eq9}-\eqref{eq10}, the iterative equations corresponding to the three-stage diffusion–permutation–diffusion PNDCC are obtained as
\begin{equation}
    \begin{aligned}
        \label{eq11}
        & \left\{ \begin{aligned}
        & {{P}_{1}}\left( k \right)=f_{D1}^{-1}\left( {{C}_{1}}\left( k \right),{{C}_{1}}\left( k-1 \right),{{C}_{1}}\left( k-2 \right),\cdots ,{{C}_{1}}\left( k-{{m}_{1}} \right),{{K}_{1}}(k) \right) \\ 
       & {{P}_{2}}\left( k \right)=f_{D2}^{-1}\left( {{C}_{2}}\left( k \right),{{C}_{2}}\left( k-1 \right),{{C}_{2}}\left( k-2 \right),\cdots ,{{C}_{2}}\left( k-{{m}_{2}} \right),{{K}_{2}}(k) \right) \\ 
      \end{aligned} \right. \\ 
       & \to \left\{ \begin{aligned}
        & P\left( k \right)=f_{D1}^{-1}\left( {{C}_{1}}\left( k \right),{{C}_{1}}\left( k-1 \right),{{C}_{1}}\left( k-2 \right),\cdots ,{{C}_{1}}\left( k-{{m}_{1}} \right),{{K}_{1}}(k) \right) \\ 
       & {{P}_{2}}\left( k \right)=f_{D2}^{-1}\left( C\left( k \right),C\left( k-1 \right),C\left( k-2 \right),\cdots ,C\left( k-{{m}_{2}} \right),{{K}_{2}}(k) \right) \\ 
      \end{aligned} \right. \\ 
       & \to \left\{ \begin{aligned}
        & P\left( k \right)=f_{D1}^{-1}\left( {{P}_{2}}\left( f_{S1}^{-1}(k) \right),{{P}_{2}}\left( f_{S1}^{-1}(k-1) \right),{{P}_{2}}\left( f_{S1}^{-1}(k-2) \right),\cdots ,{{P}_{2}}\left( f_{S1}^{-1}(k-{{m}_{1}}) \right),{{K}_{1}}(k) \right) \\ 
       & {{P}_{2}}\left( f_{S1}^{-1}(k) \right)=f_{D2}^{-1}\left( C\left( f_{S1}^{-1}(k) \right),C\left( f_{S1}^{-1}(k)-1 \right),C\left( f_{S1}^{-1}(k)-2 \right),\cdots ,C\left( f_{S1}^{-1}(k)-{{m}_{2}} \right),{{K}_{2}}(f_{S1}^{-1}(k)) \right) \\ 
       & {{P}_{2}}\left( f_{S1}^{-1}(k-1) \right)=f_{D2}^{-1}\left(\begin{aligned}
        & C\left( f_{S1}^{-1}(k-1) \right),C\left( f_{S1}^{-1}(k-1)-1 \right),C\left( f_{S1}^{-1}(k-1)-2 \right),\cdots ,\\
        &C\left( f_{S1}^{-1}(k-1)-{{m}_{2}} \right),{{K}_{2}}(f_{S1}^{-1}(k-1)) 
      \end{aligned}
        \right) \\ 
       & {{P}_{2}}\left( f_{S1}^{-1}(k-2) \right)=f_{D2}^{-1}\left(\begin{aligned}
        & C\left( f_{S1}^{-1}(k-2) \right),C\left( f_{S1}^{-1}(k-2)-1 \right),C\left( f_{S1}^{-1}(k-2)-2 \right),\cdots ,\\
        &C\left( f_{S1}^{-1}(k-2)-{{m}_{2}} \right),{{K}_{2}}(f_{S1}^{-1}(k-2)) 
      \end{aligned}
        \right) \\ 
       & \cdots  \\ 
       & {{P}_{2}}\left( f_{S1}^{-1}(k-{{m}_{1}}) \right)=f_{D2}^{-1}\left(\begin{aligned}
        & C\left( f_{S1}^{-1}(k-{{m}_{1}}) \right),C\left( f_{S1}^{-1}(k-{{m}_{1}})-1 \right),C\left( f_{S1}^{-1}(k-{{m}_{1}})-2 \right),\cdots ,\\
        &C\left( f_{S1}^{-1}(k-{{m}_{1}})-{{m}_{2}} \right),{{K}_{2}}(f_{S1}^{-1}(k-{{m}_{1}})) 
      \end{aligned}
        \right) \\ 
      \end{aligned} \right. \\ 
       & \to P\left( k \right)=f_{D1}^{-1}\left( \begin{aligned}
        & f_{D2}^{-1}\left( C\left( f_{S1}^{-1}(k) \right),C\left( f_{S1}^{-1}(k)-1 \right),C\left( f_{S1}^{-1}(k)-2 \right),\cdots ,C\left( f_{S1}^{-1}(k)-{{m}_{2}} \right),{{K}_{2}}(f_{S1}^{-1}(k)) \right), \\ 
       & f_{D2}^{-1}\left(\begin{aligned}
        & C\left( f_{S1}^{-1}(k-1) \right),C\left( f_{S1}^{-1}(k-1)-1 \right),C\left( f_{S1}^{-1}(k-1)-2 \right),\cdots ,\\
        &C\left( f_{S1}^{-1}(k-1)-{{m}_{2}} \right),{{K}_{2}}(f_{S1}^{-1}(k-1)) 
      \end{aligned}
        \right), \\ 
       & f_{D2}^{-1}\left(\begin{aligned}
        & C\left( f_{S1}^{-1}(k-2) \right),C\left( f_{S1}^{-1}(k-2)-1 \right),C\left( f_{S1}^{-1}(k-2)-2 \right),\cdots ,\\
        &C\left( f_{S1}^{-1}(k-2)-{{m}_{2}} \right),{{K}_{2}}(f_{S1}^{-1}(k-2)) 
      \end{aligned}
        \right),\\ 
        &\cdots , \\ 
       & f_{D2}^{-1}\left(\begin{aligned}
        & C\left( f_{S1}^{-1}(k-{{m}_{1}}) \right),C\left( f_{S1}^{-1}(k-{{m}_{1}})-1 \right),C\left( f_{S1}^{-1}(k-{{m}_{1}})-2 \right),\cdots ,\\
        &C\left( f_{S1}^{-1}(k-{{m}_{1}})-{{m}_{2}} \right),{{K}_{2}}(f_{S1}^{-1}(k-{{m}_{1}})) 
      \end{aligned}
        \right),{{K}_{1}}(k) \\ 
      \end{aligned} \right) \\ 
      \end{aligned} 
\end{equation}

(2) From  Fig.~\ref{fig7} and Eq.\eqref{eq11}, the iterative equations for the four-stage permutation–\allowbreak diffusion–\allowbreak permutation–\allowbreak diffusion PNDCC can be derived as
\begin{equation}
    \label{eq11-1}
    {{P}_{1}}\left( k \right)=f_{D1}^{-1}\left( \begin{aligned}
        & f_{D2}^{-1}\left(\begin{aligned}
          & C\left( f_{S1}^{-1}(k) \right),C\left( f_{S1}^{-1}(k)-1 \right),C\left( f_{S1}^{-1}(k)-2 \right),\cdots ,\\
        &C\left( f_{S1}^{-1}(k)-{{m}_{2}} \right),{{K}_{2}}\left(f_{S1}^{-1}(k)\right) 
      \end{aligned}
        \right), \\ 
       & f_{D2}^{-1}\left(\begin{aligned}
        & C\left( f_{S1}^{-1}(k-1) \right),C\left( f_{S1}^{-1}(k-1)-1 \right),C\left( f_{S1}^{-1}(k-1)-2 \right),\cdots ,\\
        &C\left( f_{S1}^{-1}(k-1)-{{m}_{2}} \right),{{K}_{2}}\left(f_{S1}^{-1}(k-1)\right) 
      \end{aligned}
        \right), \\ 
       & f_{D2}^{-1}\left(\begin{aligned}
        & C\left( f_{S1}^{-1}(k-2) \right),C\left( f_{S1}^{-1}(k-2)-1 \right),C\left( f_{S1}^{-1}(k-2)-2 \right),\cdots ,\\
        &C\left( f_{S1}^{-1}(k-2)-{{m}_{2}} \right),{{K}_{2}}\left(f_{S1}^{-1}(k-2)\right) 
      \end{aligned}
        \right),\\ 
        &\cdots , \\ 
       & f_{D2}^{-1}\left(\begin{aligned}
        & C\left( f_{S1}^{-1}(k-{{m}_{1}}) \right),C\left( f_{S1}^{-1}(k-{{m}_{1}})-1 \right),C\left( f_{S1}^{-1}(k-{{m}_{1}})-2 \right),\cdots ,\\
        &C\left( f_{S1}^{-1}(k-{{m}_{1}})-{{m}_{2}} \right),{{K}_{2}}\left(f_{S1}^{-1}(k-{{m}_{1}})\right) 
      \end{aligned}
        \right),\\
        &{{K}_{1}}(k) \\ 
      \end{aligned} \right)     
\end{equation}
Let ${{P}_{1}}(k)=P\left( f_{S0}^{{}}(k) \right)$ in the above equation, then we obtain
\begin{equation}
    \label{eq12}
    \begin{aligned}
        & P\left( f_{S0}^{{}}(k) \right)=f_{D1}^{-1}\left( \begin{aligned}
        & f_{D2}^{-1}\left(\begin{aligned}
          & C\left( f_{S1}^{-1}(k) \right),C\left( f_{S1}^{-1}(k)-1 \right),C\left( f_{S1}^{-1}(k)-2 \right),\cdots ,\\
          &C\left( f_{S1}^{-1}(k)-{{m}_{2}} \right),{{K}_{2}}\left(f_{S1}^{-1}(k)\right)
        \end{aligned}
          \right), \\ 
       & f_{D2}^{-1}\left(\begin{aligned}
        & C\left( f_{S1}^{-1}(k-1) \right),C\left( f_{S1}^{-1}(k-1)-1 \right),C\left( f_{S1}^{-1}(k-1)-2 \right),\cdots ,\\
        &C\left( f_{S1}^{-1}(k-1)-{{m}_{2}} \right),{{K}_{2}}\left(f_{S1}^{-1}(k-1)\right) 
      \end{aligned}
        \right), \\ 
       & f_{D2}^{-1}\left(\begin{aligned}
        & C\left( f_{S1}^{-1}(k-2) \right),C\left( f_{S1}^{-1}(k-2)-1 \right),C\left( f_{S1}^{-1}(k-2)-2 \right),\cdots ,\\
        &C\left( f_{S1}^{-1}(k-2)-{{m}_{2}} \right),{{K}_{2}}\left(f_{S1}^{-1}(k-2)\right) 
      \end{aligned}
        \right),\\ 
        &\cdots , \\ 
       & f_{D2}^{-1}\left(\begin{aligned}
        & C\left( f_{S1}^{-1}(k-{{m}_{1}}) \right),C\left( f_{S1}^{-1}(k-{{m}_{1}})-1 \right),C\left( f_{S1}^{-1}(k-{{m}_{1}})-2 \right),\cdots ,\\
        &C\left( f_{S1}^{-1}(k-{{m}_{1}})-{{m}_{2}} \right),{{K}_{2}}\left(f_{S1}^{-1}(k-{{m}_{1}})\right) 
      \end{aligned}
        \right),\\ 
        &{{K}_{1}}(k) \\ 
      \end{aligned} \right) \\ 
       & \to P\!\left( k \right)=f_{D1}^{-1}\!\left( \begin{aligned}
        & f_{D2}^{-1}\!\left(\begin{aligned}
        & C\!\left( f_{S1}^{-1}\!\left( f_{S0}^{-1}(k) \right) \right),C\!\left( f_{S1}^{-1}\!\left( f_{S0}^{-1}(k) \right)-1 \right),C\!\left( f_{S1}^{-1}\!\left( f_{S0}^{-1}(k) \right)-2 \right),\cdots ,\\
        &C\!\left( f_{S1}^{-1}\!\left( f_{S0}^{-1}(k) \right)-{{m}_{2}} \right),{{K}_{2}}\left( f_{S1}^{-1}\!\left( f_{S0}^{-1}(k) \right) \right) 
      \end{aligned}
        \right), \\ 
       & f_{D2}^{-1}\!\left(\begin{aligned}
        & C\!\left( f_{S1}^{-1}\!\left( f_{S0}^{-1}(k)-1 \right) \right),C\!\left( f_{S1}^{-1}\!\left( f_{S0}^{-1}(k)-1 \right)-1 \right),C\!\left( f_{S1}^{-1}\!\left( f_{S0}^{-1}(k)-1 \right)-2 \right),\cdots ,\\
        &C\!\left( f_{S1}^{-1}\!\left( f_{S0}^{-1}(k)-1 \right)-{{m}_{2}} \right),{{K}_{2}}\left( f_{S1}^{-1}\!\left( f_{S0}^{-1}(k)-1 \right) \right) 
      \end{aligned}
        \right), \\ 
       & f_{D2}^{-1}\!\left(\begin{aligned}
        & C\!\left( f_{S1}^{-1}\!\left( f_{S0}^{-1}(k)-2 \right) \right),C\!\left( f_{S1}^{-1}\!\left( f_{S0}^{-1}(k)-2 \right)-1 \right),C\!\left( f_{S1}^{-1}\!\left( f_{S0}^{-1}(k)-2 \right)-2 \right),\cdots ,\\
        &C\!\left( f_{S1}^{-1}\!\left( f_{S0}^{-1}(k)-2 \right)-{{m}_{2}} \right),{{K}_{2}}\left( f_{S1}^{-1}\!\left( f_{S0}^{-1}(k)-2 \right) \right) 
      \end{aligned}
        \right),\\ 
        &\cdots , \\ 
       & f_{D2}^{-1}\!\left( \begin{aligned}
        & C\!\left( f_{S1}^{-1}\!\left( f_{S0}^{-1}(k)-{{m}_{1}} \right) \right)\!,C\!\left( f_{S1}^{-1}\!\left( f_{S0}^{-1}(k)-{{m}_{1}} \right)-1 \right)\!,C\!\left( f_{S1}^{-1}\!\left( f_{S0}^{-1}(k)-{{m}_{1}} \right)-2 \right)\!,\cdots \!, \\ 
       & C\!\left( f_{S1}^{-1}\!\left( f_{S0}^{-1}(k)-{{m}_{1}} \right)-{{m}_{2}} \right),{{K}_{2}}\left( f_{S1}^{-1}\!\left( f_{S0}^{-1}(k)-{{m}_{1}} \right) \right) \\ 
      \end{aligned} \right),\\
      &{{K}_{1}}\left( f_{S0}^{-1}(k) \right) \\ 
      \end{aligned} \right) \\ 
      \end{aligned}      
\end{equation}

(3) From  Fig.~\ref{fig7} and Eq.\eqref{eq11}, the iterative equations for the four-stage diffusion–permutation–diffusion–permutation PNDCC can be derived as
\begin{equation}
    \label{eq12-1}
    P\left( k \right)=f_{D1}^{-1}\left( \begin{aligned}
        & f_{D2}^{-1}\left(\begin{aligned}
          & {{C}_{2}}\left( f_{S1}^{-1}(k) \right),{{C}_{2}}\left( f_{S1}^{-1}(k)-1 \right),{{C}_{2}}\left( f_{S1}^{-1}(k)-2 \right),\cdots ,\\
          &{{C}_{2}}\left( f_{S1}^{-1}(k)-{{m}_{2}} \right),{{K}_{2}}(f_{S1}^{-1}(k))
        \end{aligned}
         \right), \\ 
       & f_{D2}^{-1}\left(\begin{aligned}
        &{{C}_{2}}\left( f_{S1}^{-1}(k-1) \right),{{C}_{2}}\left( f_{S1}^{-1}(k-1)-1 \right),{{C}_{2}}\left( f_{S1}^{-1}(k-1)-2 \right),\cdots ,\\
        &{{C}_{2}}\left( f_{S1}^{-1}(k-1)-{{m}_{2}} \right),{{K}_{2}}(f_{S1}^{-1}(k-1)) 
      \end{aligned}
        \right), \\ 
       & f_{D2}^{-1}\left(\begin{aligned}
        & {{C}_{2}}\left( f_{S1}^{-1}(k-2) \right),{{C}_{2}}\left( f_{S1}^{-1}(k-2)-1 \right),{{C}_{2}}\left( f_{S1}^{-1}(k-2)-2 \right),\cdots ,\\
        &{{C}_{2}}\left( f_{S1}^{-1}(k-2)-{{m}_{2}} \right),{{K}_{2}}(f_{S1}^{-1}(k-2)) 
      \end{aligned}
        \right),\\ 
        &\cdots , \\ 
       & f_{D2}^{-1}\left(\begin{aligned}
        &  {{C}_{2}}\left( f_{S1}^{-1}(k-{{m}_{1}}) \right),{{C}_{2}}\left( f_{S1}^{-1}(k-{{m}_{1}})-1 \right),{{C}_{2}}\left( f_{S1}^{-1}(k-{{m}_{1}})-2 \right),\cdots ,\\
       &{{C}_{2}}\left( f_{S1}^{-1}(k-{{m}_{1}})-{{m}_{2}} \right),{{K}_{2}}(f_{S1}^{-1}(k-{{m}_{1}})) 
      \end{aligned}
       \right),\\ 
       &{{K}_{1}}(k) \\ 
      \end{aligned} \right)      
\end{equation}
Let ${{C}_{2}}(i)=C\left( f_{S2}^{-1}(i) \right)$ in the above equation, then we obtain
\begin{equation}
    \label{eq13}
    P\left( k \right)=f_{D1}^{-1}\left( \begin{aligned}
        & f_{D2}^{-1}\left(\begin{aligned}
          & C\left( f_{S2}^{-1}\left( f_{S1}^{-1}(k) \right) \right),C\left( f_{S2}^{-1}\left( f_{S1}^{-1}(k)-1 \right) \right),C\left( f_{S2}^{-1}\left( f_{S1}^{-1}(k)-2 \right) \right),\cdots ,\\ 
          &C\left( f_{S2}^{-1}\left( f_{S1}^{-1}(k)-{{m}_{2}} \right) \right),{{K}_{2}}(f_{S1}^{-1}(k)) 
        \end{aligned}
          \right), \\ 
       & f_{D2}^{-1}\left(\begin{aligned}
        & C\left( f_{S2}^{-1}\left( f_{S1}^{-1}(k-1) \right) \right),C\left( f_{S2}^{-1}\left( f_{S1}^{-1}(k-1)-1 \right) \right),C\left( f_{S2}^{-1}\left( f_{S1}^{-1}(k-1)-2 \right) \right),\cdots ,\\ 
        &C\left( f_{S2}^{-1}\left( f_{S1}^{-1}(k-1)-{{m}_{2}} \right) \right),{{K}_{2}}(f_{S}^{-1}(k-1)) 
      \end{aligned}
        \right), \\ 
       & f_{D2}^{-1}\left(\begin{aligned}
        & C\left( f_{S2}^{-1}\left( f_{S1}^{-1}(k-2) \right) \right),C\left( f_{S2}^{-1}\left( f_{S1}^{-1}(k-2)-1 \right) \right),C\left( f_{S2}^{-1}\left( f_{S1}^{-1}(k-2)-2 \right) \right),\cdots ,\\ 
        &C\left( f_{S2}^{-1}\left( f_{S1}^{-1}(k-2)-{{m}_{2}} \right) \right),{{K}_{2}}(f_{S1}^{-1}(k-2)) 
      \end{aligned}
        \right),\\ 
        &\cdots , \\ 
       & f_{D2}^{-1}\left(\begin{aligned}
        & C\!\left( f_{S2}^{-1}\left( f_{S1}^{-1}(k-{{m}_{1}}) \right) \right)\!,C\!\left( f_{S2}^{-1}\left( f_{S1}^{-1}(k-{{m}_{1}})-1 \right) \right)\!,C\!\left( f_{S2}^{-1}\left( f_{S1}^{-1}(k-{{m}_{1}})-2 \right) \right)\!,\cdots \!,\\ 
        &C\left( f_{S2}^{-1}\left( f_{S1}^{-1}(k-{{m}_{1}})-{{m}_{2}} \right) \right),{{K}_{2}}(f_{S1}^{-1}(k-{{m}_{1}})) 
      \end{aligned}
      \right),\\ 
      &{{K}_{1}}(k) \\ 
      \end{aligned} \right)      
\end{equation}
(4) From  Fig.~\ref{fig7} and Eq.\eqref{eq13}, the iterative equation for the five-stage permutation\allowbreak-diffusion\allowbreak-permutation\allowbreak-diffusion\allowbreak-permutation PNDCC can be derived as
\begin{equation}
    \label{eq13-1}
    {{P}_{1}}\!\left( k \right)=f_{D1}^{-1}\!\left( \begin{aligned}
        & f_{D2}^{-1}\!\left(\begin{aligned}
          & C\!\left( f_{S2}^{-1}\!\left( f_{S1}^{-1}(k) \right) \right),C\!\left( f_{S2}^{-1}\!\left( f_{S1}^{-1}(k)-1 \right) \right),C\!\left( f_{S2}^{-1}\!\left( f_{S1}^{-1}(k)-2 \right) \right),\cdots ,\\ 
          &C\!\left( f_{S2}^{-1}\!\left( f_{S1}^{-1}(k)-{{m}_{2}} \right) \right),{{K}_{2}}(f_{S1}^{-1}(k)) 
        \end{aligned}
          \right), \\ 
       & f_{D2}^{-1}\!\left(\begin{aligned}
        & C\!\left( f_{S2}^{-1}\!\left( f_{S1}^{-1}(k-1) \right) \right),C\!\left( f_{S2}^{-1}\!\left( f_{S1}^{-1}(k-1)-1 \right) \right),C\!\left( f_{S2}^{-1}\!\left( f_{S1}^{-1}(k-1)-2 \right) \right),\cdots ,\\ 
        &C\!\left( f_{S2}^{-1}\!\left( f_{S1}^{-1}(k-1)-{{m}_{2}} \right) \right),{{K}_{2}}(f_{S}^{-1}(k-1)) 
      \end{aligned}
        \right), \\ 
       & f_{D2}^{-1}\!\left(\begin{aligned}
        & C\!\left( f_{S2}^{-1}\!\left( f_{S1}^{-1}(k-2) \right) \right),C\!\left( f_{S2}^{-1}\!\left( f_{S1}^{-1}(k-2)-1 \right) \right),C\!\left( f_{S2}^{-1}\!\left( f_{S1}^{-1}(k-2)-2 \right) \right),\cdots ,\\ 
        &C\!\left( f_{S2}^{-1}\!\left( f_{S1}^{-1}(k-2)-{{m}_{2}} \right) \right),{{K}_{2}}(f_{S1}^{-1}(k-2))
      \end{aligned}
        \right),\\ 
        &\cdots , \\ 
       & f_{D2}^{-1}\!\left(\begin{aligned}
        & C\!\left( f_{S2}^{-1}\!\left( f_{S1}^{-1}(k-{{m}_{1}}) \right) \right),C\!\left( f_{S2}^{-1}\!\left( f_{S1}^{-1}(k-{{m}_{1}})-1 \right) \right),C\!\left( f_{S2}^{-1}\!\left( f_{S1}^{-1}(k-{{m}_{1}})-2 \right) \right),\cdots ,\\ 
        &C\!\left( f_{S2}^{-1}\!\left( f_{S1}^{-1}(k-{{m}_{1}})-{{m}_{2}} \right) \right),{{K}_{2}}(f_{S1}^{-1}(k-{{m}_{1}})) 
      \end{aligned}
        \right),\\ 
        &{{K}_{1}}(k) \\ 
      \end{aligned} \right)      
\end{equation}
Let ${{P}_{1}}(k)=P\left( f_{S0}^{{}}(k) \right)$ in the above equation, then we obtain
\begin{equation}
  \label{eq14}
  \begin{aligned}
    & P\!\left( f_{S0}^{{}}(k) \right)\!=\!f_{D1}^{-1}\!\left( \begin{aligned}
    & f_{D2}^{-1}\!\left(\begin{aligned}
      & C\!\left( f_{S2}^{-1}\!\left( f_{S1}^{-1}(k) \right) \right),C\!\left( f_{S2}^{-1}\!\left( f_{S1}^{-1}(k)-1 \right) \right),C\!\left( f_{S2}^{-1}\!\left( f_{S1}^{-1}(k)-2 \right) \right),\cdots ,\\ 
      &C\!\left( f_{S2}^{-1}\!\left( f_{S1}^{-1}(k)-{{m}_{2}} \right) \right),{{K}_{2}}(f_{S1}^{-1}(k)) 
    \end{aligned}
      \right), \\ 
   & f_{D2}^{-1}\!\left(\begin{aligned}
    &  C\!\left( f_{S2}^{-1}\!\left( f_{S1}^{-1}(k-1) \right) \right),C\!\left( f_{S2}^{-1}\!\left( f_{S1}^{-1}(k-1)-1 \right) \right),C\!\left( f_{S2}^{-1}\!\left( f_{S1}^{-1}(k-1)-2 \right) \right),\cdots ,\\ 
    &C\!\left( f_{S2}^{-1}\!\left( f_{S1}^{-1}(k-1)-{{m}_{2}} \right) \right),{{K}_{2}}(f_{S}^{-1}(k-1)) 
  \end{aligned}
 \right), \\ 
   & f_{D2}^{-1}\!\left(\begin{aligned}
    &  C\!\left( f_{S2}^{-1}\!\left( f_{S1}^{-1}(k-2) \right) \right)\!,C\!\left( f_{S2}^{-1}\!\left( f_{S1}^{-1}(k-2)-1 \right) \right)\!,C\!\left( f_{S2}^{-1}\!\left( f_{S1}^{-1}(k-2)-2 \right) \right)\!,\cdots \!,\\ 
    &C\!\left( f_{S2}^{-1}\!\left( f_{S1}^{-1}(k-2)-{{m}_{2}} \right) \right),{{K}_{2}}(f_{S1}^{-1}(k-2)) 
  \end{aligned}
 \right)\!,\cdots \!, \\ 
   & f_{D2}^{-1}\!\left(\!\begin{aligned}
    &  C\!\left( f_{S2}^{-1}\!\left( f_{S1}^{-1}(k-{{m}_{1}})\! \right) \!\right)\!,C\!\left( f_{S2}^{-1}\!\left( f_{S1}^{-1}(k-{{m}_{1}})-1 \right)\! \right)\!,C\!\left( f_{S2}^{-1}\!\left( f_{S1}^{-1}(k-{{m}_{1}})-2 \right)\! \right)\!,\cdots \!,\\ 
    &C\!\left( f_{S2}^{-1}\!\left( f_{S1}^{-1}(k-{{m}_{1}})-{{m}_{2}} \right) \right),{{K}_{2}}(f_{S1}^{-1}(k-{{m}_{1}})) 
  \end{aligned}\!
 \right)\!,\\ 
 &{{K}_{1}}(k) \\ 
  \end{aligned} \right) \\ 
   & \to P(k)=f_{D1}^{-1}\!\left( \begin{aligned}
    & f_{D2}^{-1}\!\left( \begin{aligned}
    & C\!\left( f_{S2}^{-1}\!\left( f_{S1}^{-1}\!\left( f_{S0}^{-1}\!\left( k \right) \right) \right) \right),C\!\left( f_{S2}^{-1}\!\left( f_{S1}^{-1}\!\left( f_{S0}^{-1}\!\left( k \right) \right)-1 \right) \right),\\ 
    &C\!\left( f_{S2}^{-1}\!\left( f_{S1}^{-1}\!\left( f_{S0}^{-1}\!\left( k \right) \right)-2 \right) \right),\cdots , \\ 
   & C\!\left( f_{S2}^{-1}\!\left( f_{S1}^{-1}\!\left( f_{S0}^{-1}\!\left( k \right) \right)-{{m}_{2}} \right) \right),{{K}_{2}}\!\left( f_{S1}^{-1}\!\left( f_{S0}^{-1}\!\left( k \right) \right) \right) \\ 
  \end{aligned} \right), \\ 
   & f_{D2}^{-1}\!\left( \begin{aligned}
    & C\!\left( f_{S2}^{-1}\!\left( f_{S1}^{-1}\!\left( f_{S0}^{-1}\!\left( k \right)-1 \right) \right) \right),C\!\left( f_{S2}^{-1}\!\left( f_{S1}^{-1}\!\left( f_{S0}^{-1}\!\left( k \right)-1 \right)-1 \right) \right),\\ 
    &C\!\left( f_{S2}^{-1}\!\left( f_{S1}^{-1}\!\left( f_{S0}^{-1}\!\left( k \right)-1 \right)-2 \right) \right),\cdots , \\ 
   & C\!\left( f_{S2}^{-1}\!\left( f_{S1}^{-1}\!\left( f_{S0}^{-1}\!\left( k \right)-1 \right)-{{m}_{2}} \right) \right),{{K}_{2}}\!\left( f_{S}^{-1}\!\left( f_{S0}^{-1}\!\left( k \right)-1 \right) \right) \\ 
  \end{aligned} \right), \\ 
   & f_{D2}^{-1}\!\left( \begin{aligned}
    & C\!\left( f_{S2}^{-1}\!\left( f_{S1}^{-1}\!\left( f_{S0}^{-1}\!\left( k \right)-2 \right) \right) \right),C\!\left( f_{S2}^{-1}\!\left( f_{S1}^{-1}\!\left( f_{S0}^{-1}\!\left( k \right)-2 \right)-1 \right) \right),\\ 
    &C\!\left( f_{S2}^{-1}\!\left( f_{S1}^{-1}\!\left( f_{S0}^{-1}\!\left( k \right)-2 \right)-2 \right) \right),\cdots , \\ 
   & C\!\left( f_{S2}^{-1}\!\left( f_{S1}^{-1}\!\left( f_{S0}^{-1}\!\left( k \right)-2 \right)-{{m}_{2}} \right) \right),{{K}_{2}}\!\left( f_{S1}^{-1}\!\left( f_{S0}^{-1}\!\left( k \right)-2 \right) \right) \\ 
  \end{aligned} \right),\cdots , \\ 
   & f_{D2}^{-1}\!\left( \begin{aligned}
    & C\!\left( f_{S2}^{-1}\!\left( f_{S1}^{-1}\!\left( f_{S0}^{-1}\!\left( k \right)-{{m}_{1}} \right) \right) \right),C\!\left( f_{S2}^{-1}\!\left( f_{S1}^{-1}\!\left( f_{S0}^{-1}\!\left( k \right)-{{m}_{1}} \right)-1 \right) \right), \\ 
   & C\!\left( f_{S2}^{-1}\!\left( f_{S1}^{-1}\!\left( f_{S0}^{-1}\!\left( k \right)-{{m}_{1}} \right)-2 \right) \right),\cdots ,\\ 
   &C\!\left( f_{S2}^{-1}\!\left( f_{S1}^{-1}\!\left( f_{S0}^{-1}\!\left( k \right)-{{m}_{1}} \right)-{{m}_{2}} \right) \right),{{K}_{2}}\!\left( f_{S1}^{-1}\!\left( f_{S0}^{-1}\!\left( k \right)-{{m}_{1}} \right) \right) \\ 
  \end{aligned} \right),{{K}_{1}}\!\left( f_{S0}^{-1}\!\left( k \right) \right) \\ 
  \end{aligned} \right) \\ 
  \end{aligned}  
\end{equation}
The derivations of the iterative equations for the remaining multi-stage PNDCCs proceed similarly and are therefore omitted.

Note that the generalized iterative equation in Eqs.~\eqref{eq11}-\eqref{eq14} is essentially a complete decryption machine: feeding ciphertext at the input yields the corresponding plaintext. Images have two dimensions, intensity and position; “Complete decryption machine” means recovering both dimensions. In chain attacks, one first recovers positional permutations; given these, the diffusion can be further recovered by solving the resulting systems of equations.

\subsection{Associated Position Expression for Multi-Stage PNDCC}
The above analysis shows that a positional decryption machine is crucial for chain attacks. How do we extract it from the complete decryption machine in Eqs.~\eqref{eq11}-\eqref{eq14}? Equivalently, how do we obtain the associated position expression? The method is to retain only the positional mappings and permutations in Eqs.~\eqref{eq11}-\eqref{eq14}, ignoring diffusion.

Based on Eqs.~\eqref{eq11}-\eqref{eq14}, the associated-position expression for the three-stage diffusion–permutation–diffusion PNDCC is given by
\begin{equation}
    \label{eq15}
    \left\{ P\left( k \right) \right\}\leftarrow \left\{ \begin{aligned}
        & C\left( f_{S1}^{-1}(k) \right),C\left( f_{S1}^{-1}(k)-1 \right),C\left( f_{S1}^{-1}(k)-2 \right),\cdots ,C\left( f_{S1}^{-1}(k)-{{m}_{2}} \right), \\ 
       & C\left( f_{S1}^{-1}(k-1) \right),C\left( f_{S1}^{-1}(k-1)-1 \right),C\left( f_{S1}^{-1}(k-1)-2 \right),\cdots ,C\left( f_{S1}^{-1}(k-1)-{{m}_{2}} \right), \\ 
       & C\left( f_{S1}^{-1}(k-2) \right),C\left( f_{S1}^{-1}(k-2)-1 \right),C\left( f_{S1}^{-1}(k-2)-2 \right),\cdots ,C\left( f_{S1}^{-1}(k-2)-{{m}_{2}} \right),\\ 
       &\cdots , \\ 
       & C\left( f_{S1}^{-1}(k-{{m}_{1}}) \right),C\left( f_{S1}^{-1}(k-{{m}_{1}})-1 \right),C\left( f_{S1}^{-1}(k-{{m}_{1}})-2 \right),\cdots ,C\left( f_{S1}^{-1}(k-{{m}_{1}})-{{m}_{2}} \right) \\ 
      \end{aligned} \right\}      
  \end{equation}
Similarly, the associated-position expression for the four-stage permutation–diffusion–permutation–diffusion PNDCC is given by
\begin{equation}
    \label{eq16}
    \left\{ P\left( k \right) \right\}\leftarrow \left\{ \begin{aligned}
        & C\left( f_{S1}^{-1}\left( f_{S0}^{-1}(k) \right) \right),C\left( f_{S1}^{-1}\left( f_{S0}^{-1}(k) \right)-1 \right),C\left( f_{S1}^{-1}\left( f_{S0}^{-1}(k) \right)-2 \right),\cdots ,\\ 
        &C\left( f_{S1}^{-1}\left( f_{S0}^{-1}(k) \right)-{{m}_{2}} \right), \\ 
       & C\left( f_{S1}^{-1}\left( f_{S0}^{-1}(k)-1 \right) \right),C\left( f_{S1}^{-1}\left( f_{S0}^{-1}(k)-1 \right)-1 \right),C\left( f_{S1}^{-1}\left( f_{S0}^{-1}(k)-1 \right)-2 \right),\cdots ,\\ 
       &C\left( f_{S1}^{-1}\left( f_{S0}^{-1}(k)-1 \right)-{{m}_{2}} \right), \\ 
       & C\left( f_{S1}^{-1}\left( f_{S0}^{-1}(k)-2 \right) \right),C\left( f_{S1}^{-1}\left( f_{S0}^{-1}(k)-2 \right)-1 \right),C\left( f_{S1}^{-1}\left( f_{S0}^{-1}(k)-2 \right)-2 \right),\cdots ,\\ 
       &C\left( f_{S1}^{-1}\left( f_{S0}^{-1}(k)-2 \right)-{{m}_{2}} \right),\\ 
       &\cdots , \\ 
       & C\left( f_{S1}^{-1}\left( f_{S0}^{-1}(k)-{{m}_{1}} \right) \right),C\left( f_{S1}^{-1}\left( f_{S0}^{-1}(k)-{{m}_{1}} \right)-1 \right),C\left( f_{S1}^{-1}\left( f_{S0}^{-1}(k)-{{m}_{1}} \right)-2 \right),\cdots ,\\ 
       &C\left( f_{S1}^{-1}\left( f_{S0}^{-1}(k)-{{m}_{1}} \right)-{{m}_{2}} \right) \\ 
      \end{aligned} \right\}      
  \end{equation}
Likewise, the associated-position expression for the four-stage diffusion–permutation–diffusion–permutation PNDCC is given by
\begin{equation}
    \label{eq17}
    \left\{ P\left( k \right) \right\}\leftarrow \left\{ \begin{aligned}
        & C\left( f_{S2}^{-1}\left( f_{S1}^{-1}(k) \right) \right),C\left( f_{S2}^{-1}\left( f_{S1}^{-1}(k)-1 \right) \right),C\left( f_{S2}^{-1}\left( f_{S1}^{-1}(k)-2 \right) \right),\cdots ,\\ 
        &C\left( f_{S2}^{-1}\left( f_{S1}^{-1}(k)-{{m}_{2}} \right) \right), \\ 
       & C\left( f_{S2}^{-1}\left( f_{S1}^{-1}(k-1) \right) \right),C\left( f_{S2}^{-1}\left( f_{S1}^{-1}(k-1)-1 \right) \right),C\left( f_{S2}^{-1}\left( f_{S1}^{-1}(k-1)-2 \right) \right),\cdots ,\\ 
       &C\left( f_{S2}^{-1}\left( f_{S1}^{-1}(k-1)-{{m}_{2}} \right) \right), \\ 
       & C\left( f_{S2}^{-1}\left( f_{S1}^{-1}(k-2) \right) \right),C\left( f_{S2}^{-1}\left( f_{S1}^{-1}(k-2)-1 \right) \right),C\left( f_{S2}^{-1}\left( f_{S1}^{-1}(k-2)-2 \right) \right),\cdots ,\\ 
       &C\left( f_{S2}^{-1}\left( f_{S1}^{-1}(k-2)-{{m}_{2}} \right) \right),\\ 
       &\cdots , \\ 
       & C\left( f_{S2}^{-1}\left( f_{S1}^{-1}(k-{{m}_{1}}) \right) \right),C\left( f_{S2}^{-1}\left( f_{S1}^{-1}(k-{{m}_{1}})-1 \right) \right),C\left( f_{S2}^{-1}\left( f_{S1}^{-1}(k-{{m}_{1}})-2 \right) \right),\cdots ,\\ 
       &C\left( f_{S2}^{-1}\left( f_{S1}^{-1}(k-{{m}_{1}})-{{m}_{2}} \right) \right) \\ 
      \end{aligned} \right\}      
  \end{equation}
Finally, the associated-position expression for the five-stage permutation–diffusion–permutation–diffusion–permutation PNDCC is given by
\begin{equation}
  \label{eq18}
  \left\{ P(k) \right\}\!\leftarrow\! \left\{ \begin{aligned}
    & C\!\left( f_{S2}^{-\!1}\!\left( f_{S1}^{-\!1}\!\left( f_{S0}^{-\!1}\!\left( k \right) \right) \right) \right),C\!\left( f_{S2}^{-\!1}\!\left( f_{S1}^{-\!1}\!\left( f_{S0}^{-\!1}\!\left( k \right) \right)\!-\!1 \right) \right),C\!\left( f_{S2}^{-\!1}\!\left( f_{S1}^{-\!1}\!\left( f_{S0}^{-\!1}\!\left( k \right) \right)\!-\!2 \right) \right),\cdots ,\\ 
    &C\!\left( f_{S2}^{-\!1}\!\left( f_{S1}^{-\!1}\!\left( f_{S0}^{-\!1}\!\left( k \right) \right)\!-\!{{m}_{2}} \right) \right), \\ 
   & C\!\left( f_{S2}^{-\!1}\!\left( f_{S1}^{-\!1}\!\left( f_{S0}^{-\!1}\!\left( k \right)\!-\!1 \right) \right) \right),C\!\left( f_{S2}^{-\!1}\!\left( f_{S1}^{-\!1}\!\left( f_{S0}^{\!-\!1}\!\left( k \right)\!-\!1 \right)\!-\!1 \right) \right),C\!\left( f_{S2}^{-\!1}\!\left( f_{S1}^{-\!1}\!\left( f_{S0}^{-\!1}\!\left( k \right)\!-\!1 \right)\!-\!2 \right) \right),\cdots ,\\ 
   &C\!\left( f_{S2}^{-\!1}\!\left( f_{S1}^{-\!1}\!\left( f_{S0}^{-\!1}\!\left( k \right)\!-\!1 \right)\!-\!{{m}_{2}} \right) \right), \\ 
   & C\!\left( f_{S2}^{-\!1}\!\left( f_{S1}^{-\!1}\!\left( f_{S0}^{-\!1}\!\left( k \right)\!-\!2 \right) \right) \right),C\!\left( f_{S2}^{-\!1}\!\left( f_{S1}^{-\!1}\!\left( f_{S0}^{-\!1}\!\left( k \right)\!-\!2 \right)\!-\!1 \right) \right),C\!\left( f_{S2}^{-\!1}\!\left( f_{S1}^{-\!1}\!\left( f_{S0}^{-\!1}\!\left( k \right)\!-\!2 \right)\!-\!2 \right) \right),\cdots ,\\ 
   &C\!\left( f_{S2}^{-\!1}\!\left( f_{S1}^{-\!1}\!\left( f_{S0}^{-\!1}\!\left( k \right)\!-\!2 \right)\!-\!{{m}_{2}} \right) \right), \\ 
   &\cdots ,\\ 
   & C\!\left( f_{S2}^{-\!1}\!\left( f_{S1}^{-\!1}\!\left( f_{S0}^{-\!1}\!\left( k \right)\!-\!{{m}_{1}} \right) \right) \right),C\!\left( f_{S2}^{-\!1}\!\left( f_{S1}^{-\!1}\!\left( f_{S0}^{-\!1}\!\left( k \right)\!-\!{{m}_{1}} \right)\!-\!1 \right) \right),C\!\left( f_{S2}^{-\!1}\!\left( f_{S1}^{-\!1}\!\left( f_{S0}^{-\!1}\!\left( k \right)\!-\!{{m}_{1}} \right)\!-\!2 \right) \right),\\ 
   &\cdots ,C\!\left( f_{S2}^{-\!1}\!\left( f_{S1}^{-\!1}\!\left( f_{S0}^{-\!1}\!\left( k \right)\!-\!{{m}_{1}} \right)\!-\!{{m}_{2}} \right) \right) \\ 
  \end{aligned} \right\}  
\end{equation}

By the same method, one can derive the associated position expressions for an arbitrary number of stages; due to space limits, they are omitted here. The left arrow indicates that plaintext positions are determined from ciphertext via chosen-ciphertext attacks. According to Eq.~\eqref{eq15}-\eqref{eq18}, an attacker can temporarily obtain the positional decryption capability: inject a perturbation at each ciphertext position in turn, and record the perturbation responses on the plaintext side, thereby determining all associated positions between plaintext and ciphertext.

\subsection{Importance of the Generalized Iterative Equation and Associated Position Expression for Multi-Stage PNDCC}
The importance of the generalized iterative equation and the associated position expression for multi-stage PNDCC is mainly reflected in the following aspects:

(1) It must be emphasized that, since multi-stage PNDCC contains ciphertext-feedback terms but no plaintext-delay terms, the closed forms in Eqs.~\eqref{eq14}-\eqref{eq18} can be obtained only under this specific prerequisite; otherwise, closed-form expressions of the type shown in Eqs.~\eqref{eq14}-\eqref{eq18} cannot be derived.

(2) From an attack perspective, Eqs.~\eqref{eq14}-\eqref{eq18} justify first recovering all permutations and then all diffusions, thereby providing a theoretical basis for chain attacks. Using chain attacks to break the multi-stage PNDCC in Fig.~\ref{fig7} 
is comparatively simple and of low complexity, making it more efficient than many conventional cryptanalytic methods.

(3) According to Eq.~\eqref{eq18}, one can theoretically prove the existence of a complete chain of length $M\times N$, which establishes the prerequisite for the chain attack introduced in Sec.~\ref{sec:chained-attack}.

\section{Chain Attack on Multi-Stage PNDCC and Its Implementation}
\label{sec:chained-attack}
The core principle of the chain attack is that the pixel positions obtained by chosen-ciphertext attack are permuted positions, so the absolute positions $k=1,2,\ldots,M\times N$ cannot be obtained directly. However, if one can find a chain head and a complete chain of length $M\times N$, then starting from the chain head and traversing node by node yields an order that exactly matches the absolute positions $k=1,2,\ldots,M\times N$. Only when both the absolute positions and the permuted positions are obtained and compared can the permutation key be recovered. Conversely, if the chain head cannot be found and a complete chain of length $M\times N$ cannot be formed, the absolute positions cannot be determined, and the permutation key cannot be broken.

\subsection{Basic Concepts of the Chain}
In Sec.~\ref{sec:general-iter-eq}, we derived the general iterative form of the “positional decryption machine” for 5-stage PNDCC, given by Eq.~\eqref{eq18}. The attacker’s goal is to temporarily obtain the positional decryption capability provided by Eq.~\eqref{eq18}. Based on this, under a chosen-ciphertext attack, the attacker injects a perturbation at each ciphertext terminal in turn and records the perturbation responses on the plaintext side, thereby establishing the positional correspondence between the plaintext set and the ciphertext set. Finally, the permutation key is recovered using the chain attack.

However, a prerequisite for recovering the permutation key via the chain attack is that Eq.~\eqref{eq18} must contain a complete chain of length $M\times N$, hereafter simply referred to as a chain. If one can prove from Eq.~\eqref{eq18} that such a chain indeed exists, then the attacker can locate this chain via a chosen-ciphertext attack and thus recover the corresponding permutation key. Under this premise, we first need to prove theoretically that Eq.~\eqref{eq18} contains a complete chain of length $M\times N$.

\begin{definition}[Chain]
  \label{def1}
  A chain consists of multiple nodes, each representing a set. If the sets of two adjacent nodes have a nonempty intersection, then there exists a link between them, forming a chain. The chain head is the starting node, the chain tail is the terminal node, and the nodes between the head and tail are intermediate nodes.
  \end{definition}

Note that the original Def.~\ref{def1} of a chain is a basic definition and does not provide a concrete construction method. In practice, a chain can be constructed by taking intersections. For example, for a complete chain of length $M\times N$, every element except those in the last set corresponding to the chain tail is included by at least one intersection.

\begin{theorem}[Existence Theorem of the Chain]
  \label{th1}
  For the positional correspondence between the plaintext set and the ciphertext set represented by Eq.~\eqref{eq18}, there must exist a complete chain of length $M\times N$.
  \end{theorem}

  \begin{proof}
    For brevity, we prove the claim using the associated position expression in Eq.~\eqref{eq18} for the five-stage permutation-diffusion-permutation-diffusion-permutation PNDCC, and let $m_{1}=m_{2}=1$. It suffices to expand Eq.~\eqref{eq18} for $k=1,2,\ldots,M\times N$ to observe a complete chain of length $M\times N$. Note that this proof approach also holds in the general case. In fact, setting $m_{1}=m_{2}=1$ in Eq.~\eqref{eq18} and transforming yields
\begin{equation}
  \label{eq19}
\{ P\!\left( f_{S0}^{{}}\left( k \right) \right) \} \!\leftarrow\! \left\{\!
   C\!\left( f_{S2}^{-1}\left( f_{S1}^{-1}(k) \right) \right)\!,C\!\left( f_{S2}^{-1}\left( f_{S1}^{-1}(k)-1 \right) \right)\!,C\!\left( f_{S2}^{-1}\left( f_{S1}^{-1}(k-1) \right) \right)\!,C\!\left( f_{S2}^{-1}\left( f_{S1}^{-1}(k-1)-1\right) \right)    
   \!\right\}.
\end{equation}

Expanding Eq.~\eqref{eq19} for $k=1,2,\ldots,M\times N$, we find that, except for the elements $C(f_{S2}^{-1}(f_{S1}^{-1}(M\times N))),C(f_{S2}^{-1}(f_{S1}^{-1}(M\times N)-1))$ in the last set corresponding to the chain tail, all other elements in the chain are included by at least one intersection. Thus, we obtain a complete chain of length $M\times N$, as illustrated in Fig.~\ref{fig8}.
\begin{figure}[htbp] 
  \centering 
  \includegraphics[width=0.99\textwidth]{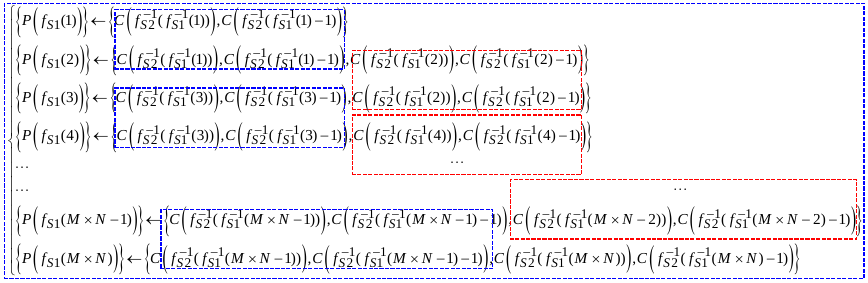} 
  \caption{Schematic of a complete chain of length $M\times N$ derived from Eq.~\eqref{eq19}.} 
  \label{fig8} 
\end{figure}
\end{proof}

\subsection{Basic Method of the Chain Attack}
In this section, we continue to use the five-stage permutation\allowbreak-diffusion\allowbreak-permutation\allowbreak-diffusion\allowbreak-permutation PNDCC as an example to illustrate the basic method of the chain attack, which generalizes to other multi-stage cases without loss of generality. Decrypting associated positions via the positional decryption machine means directly computing the associated positions using Eq.~\eqref{eq18}; by Thm.~\ref{th1}, one can find a complete chain of length $M\times N$, and this process is termed decryption. By contrast, when the attacker temporarily obtains the positional decryption capability of Eq.~\eqref{eq18} and, under a chosen-ciphertext attack, injects a perturbation at each ciphertext terminal in turn and records the corresponding perturbation responses on the plaintext side to establish the positional association between plaintext and ciphertext pixels, this process is termed cryptanalysis.

(1) Basic principle of decryption

For the five-stage PNDCC, set $m_{1}=m_{2}=1$ in Eq.~\eqref{eq14}, yielding
\begin{equation}
  \label{eq20}
P(k)\!=\!f_{D1}^{-1}\left(\! \begin{aligned}
  & f_{D2}^{-1}\left(
 C\left( f_{S2}^{-1}\left( f_{S1}^{-1}\left( f_{S0}^{-1}\left( k \right) \right) \right) \right),C\left( f_{S2}^{-1}\left( f_{S1}^{-1}\left( f_{S0}^{-1}\left( k \right) \right)-1 \right) \right),{{K}_{2}}\left( f_{S1}^{-1}\left( f_{S0}^{-1}\left( k \right) \right) \right) \right), \\ 
 & f_{D2}^{-1}\left(
 C\left( f_{S2}^{-1}\left( f_{S1}^{-1}\left( f_{S0}^{-1}\left( k \right)-1 \right) \right) \right)\!,C\left( f_{S2}^{-1}\left( f_{S1}^{-1}\left( f_{S0}^{-1}\left( k \right)-1 \right)-1 \right) \right)\!,{{K}_{2}}\left( f_{S1}^{-1}\left( f_{S0}^{-1}\left( k \right)-1 \right) \right)  \right)\!,\\ 
&{{K}_{1}}\left( f_{S0}^{-1}\left( k \right) \right) \\ 
\end{aligned} \!\right)\!,
\end{equation}
where $k=1,2,\ldots,M\times N$, and $M\times N$ is the image size. For simplicity without loss of generality, let $M\times N=9$, the plaintext be $P(1),P(2),\ldots,P(9)$, the ciphertext be $C(1),C(2),\ldots,C(9)$, and set the pre-, mid-, and post-permutations in Eq.~\eqref{eq20} as
\begin{equation}
  \label{eq21}
  \left\{ \begin{aligned}
  & f_{S0}^{-1}:(1,4),(2,5),(3,2),(4,9),(5,8),(6,7),(7,6),(8,3),(9,1) \\ 
 & f_{S1}^{-1}:(1,9),(2,2),(3,4),(4,5),(5,1),(6,6),(7,3),(8,8),(9,7) \\ 
 & f_{S2}^{-1}:(1,3),(2,6),(3,7),(4,8),(5,2),(6,1),(7,9),(8,4),(9,5) \\ 
\end{aligned} \right..
\end{equation}
In addition, the general form of the pre-diffusion iteration in Eq.~\eqref{eq20} is
\begin{equation}
  \label{eq22}
{{P}_{1}}\left( k \right)=f_{D1}^{-1}\left( {{C}_{1}}\left( k \right),{{C}_{1}}\left( k-1 \right),{{K}_{1}}(k) \right),
\end{equation}
with the corresponding pre-diffusion positional decryption machine
\begin{equation}
  \label{eq23}
\begin{aligned}
  & \left\{ \begin{aligned}
  & \left\{ {{P}_{1}}\left( k \right) \right\}\leftarrow \left\{ {{C}_{1}}\left( k \right),{{C}_{1}}\left( k-1 \right) \right\} \\ 
 & \left\{ {{P}_{1}}\left( k+1 \right) \right\}\leftarrow \left\{ {{C}_{1}}\left( k+1 \right),{{C}_{1}}\left( k \right) \right\} \\ 
\end{aligned} \right.  \to \left\{ {{P}_{1}}\left( k \right),{{P}_{1}}\left( k+1 \right) \right\}\leftarrow \left\{ {{C}_{1}}\left( k \right) \right\} .
\end{aligned}
\end{equation}
The general form of the post-diffusion iteration is
\begin{equation}
  \label{eq24}
{{P}_{2}}\left( k \right)=f_{D2}^{-1}\left( {{C}_{2}}\left( k \right),{{C}_{2}}\left( k-1 \right),{{K}_{2}}(k) \right),
\end{equation}
with the corresponding post-diffusion positional decryption machine
\begin{equation}
  \label{eq25}
\begin{aligned}
  & \left\{ \begin{aligned}
  & \left\{ {{P}_{2}}\left( k \right) \right\}\leftarrow \left\{ {{C}_{2}}\left( k \right),{{C}_{2}}\left( k-1 \right) \right\} \\ 
 & \left\{ {{P}_{2}}\left( k+1 \right) \right\}\leftarrow \left\{ {{C}_{2}}\left( k+1 \right),{{C}_{2}}\left( k \right) \right\} \\ 
\end{aligned} \right.  \to \left\{ {{P}_{2}}\left( k \right),{{P}_{2}}\left( k+1 \right) \right\}\leftarrow \left\{ {{C}_{2}}\left( k \right) \right\} .
\end{aligned}
\end{equation}

From Eqs.~\eqref{eq20}-\eqref{eq25}, the decryption result is as shown in Fig.~\ref{fig9}. It can be seen that the decryption machine forms a complete chain of length 9, from which the following six linking rules can be summarized:

1) 
When a set of four ciphertext elements is adjacent to another set of four ciphertext elements, their intersection is allowed to contain only two ciphertext elements.

2) 
When a set of three ciphertext elements is adjacent to a set of four ciphertext elements, their intersection is allowed to contain only two ciphertext elements.

3) 
When a set of two ciphertext elements is adjacent to a set of four ciphertext elements, their intersection is allowed to contain only two ciphertext elements.

4) 
When a set of two ciphertext elements is adjacent to a set of three ciphertext elements, their intersection may contain either two ciphertext elements or only one; the correct case must be confirmed by exhaustive trial and error.

5) 
When a set of three ciphertext elements is adjacent to another set of three ciphertext elements, there are two possible outcomes for their intersection: it may contain only one ciphertext element or only two. Which is correct must be determined by trial and error.

6) 
The last position cannot be handled by taking intersections with a subsequent item and must be discussed separately. If the last plaintext-ciphertext association corresponds to four ciphertext elements, the correct last position can be obtained directly by taking intersections and set differences. If it corresponds to two or three ciphertext elements, exhaustive trial and error is required.

\begin{figure}[htbp] 
  \centering 
  \includegraphics[width=0.45\textwidth]{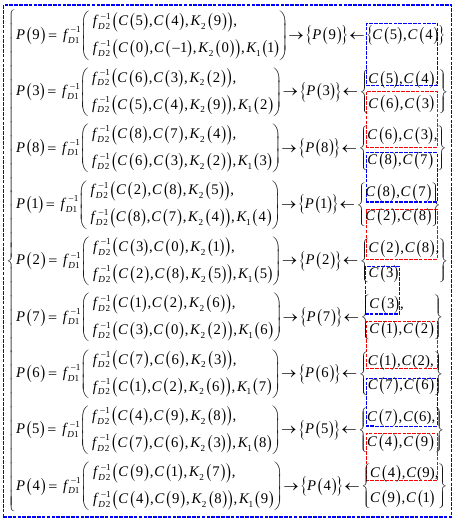} 
  \caption{A complete chain of length 9 formed by the decryption machine.} 
  \label{fig9} 
\end{figure}

(2) Basic method for recovering the pre-permutation $f_{S0}^{-1}$ via the chain attack

From Eqs.~\eqref{eq20}-\eqref{eq25}, the corresponding decryption-machine block diagram is shown in Fig.~\ref{fig10}. Its main features are as follows:

1) 
After the ciphertext $C$ passes through the post-permutation $f_{S2}^{-1}$ and post-diffusion $f_{D2}^{-1}$, a complete chain of length 9 is formed at  $P_{2}^{{}}$, and within this chain, the number of ciphertext elements inside each dashed circle is 2.

2) 
The complete chain at $P_{2}^{{}}$ is broken at $C_{1}$ after the mid-permutation $f_{S1}^{-1}$. In other words, if we can recover $f_{S1}^{-1}$, we can reconnect the broken chain at $C_{1}$ into a complete chain of length 9.

3) 
The broken chain at $C_{1}$ becomes a complete chain of length 9 again at $P_{1}$ after the pre-diffusion $f_{D1}^{-1}$, and within this chain the number of ciphertext elements inside each dashed circle is typically 4.

4) 
The complete chain at $P_{1}$ is broken at the output node $P$ by the pre-permutation $f_{S0}^{-1}$. In other words, if we can recover $f_{S0}^{-1}$, we can reconnect the broken chain at $P$ into a complete chain of length 9.

5) 
Each diffusion creates a complete chain; each permutation breaks the complete chain. In the chain formed after post-diffusion, the dashed circles contain 2 ciphertext elements, but after one more pre-diffusion, the number doubles to typically 4. According to this rule, when attacking the mid-permutation, we must reduce the number of ciphertext elements inside dashed circles from 4 to 2 by taking unions over plaintext and intersections over ciphertext. This will be detailed later.

\begin{figure*}[htbp] 
  \centering 
  \includegraphics[width=0.9\textwidth]{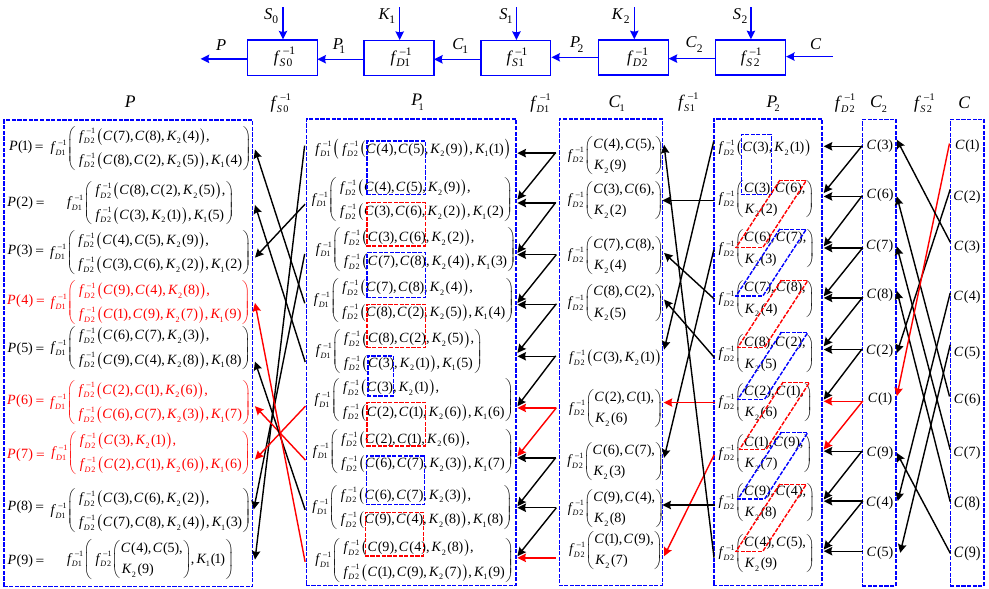} 
  \caption{Block diagram of the decryption machine and illustration of the change in $\{P(4),P(6),P(7)\}$ caused by a perturbation injected at $\left\{ C(1) \right\}$ under chosen-ciphertext attack.} 
  \label{fig10} 
\end{figure*}

Under a chosen-ciphertext attack, feed a monochromatic ciphertext $C^{(0)}=V_{1}$ into the decryption machine in Fig.~\ref{fig10} to obtain the corresponding plaintext $P^{(0)}$. Similarly, input the $M\times N$ basis images ${{e}^{(i)}}$ (${{e}^{(i)}}$ has $V_{2}$ at position $i$ and $V_{1}$ elsewhere) into the decryption machine to obtain all $P^{(i)} (i=1,2,\ldots,M\times N)$. Let $\Delta {{P}^{(i)}}\text{=}{{P}^{(i)}}-{{P}^{(0)}}(i=1,2,\cdots ,M\times N)$, where $M\times N=9$, and record the positions of nonzero elements in $\Delta {{P}^{(i)}}$. Then
\begin{equation}
  \label{eq26}
\left\{ \begin{aligned}
  & \Delta {{P}^{(1)}}(4)\ne 0,\Delta {{P}^{(1)}}(6)\ne 0,\Delta {{P}^{(1)}}(7)\ne 0 \\ 
 & \Delta {{P}^{(1)}}(1)\ne 0,\Delta {{P}^{(2)}}(2)\ne 0,\Delta {{P}^{(2)}}(6)\ne 0,\Delta {{P}^{(2)}}(7)\ne 0 \\ 
 & \Delta {{P}^{(3)}}(2)\ne 0,\Delta {{P}^{(3)}}(3)\ne 0,\Delta {{P}^{(3)}}(7)\ne 0,\Delta {{P}^{(3)}}(8)\ne 0 \\ 
 & \Delta {{P}^{(4)}}(3)\ne 0,\Delta {{P}^{(4)}}(4)\ne 0,\Delta {{P}^{(4)}}(5)\ne 0,\Delta {{P}^{(4)}}(9)\ne 0 \\ 
 & \Delta {{P}^{(6)}}(3)\ne 0,\Delta {{P}^{(6)}}(9)\ne 0 \\ 
 & \Delta {{P}^{(5)}}(3)\ne 0,\Delta {{P}^{(5)}}(5)\ne 0,\Delta {{P}^{(5)}}(6)\ne 0,\Delta {{P}^{(5)}}(8)\ne 0 \\ 
 & \Delta {{P}^{(7)}}(1)\ne 0,\Delta {{P}^{(5)}}(5)\ne 0,\Delta {{P}^{(5)}}(6)\ne 0,\Delta {{P}^{(5)}}(8)\ne 0 \\ 
 & \Delta {{P}^{(8)}}(1)\ne 0,\Delta {{P}^{(8)}}(2)\ne 0,\Delta {{P}^{(8)}}(8)\ne 0 \\ 
 & \Delta {{P}^{(9)}}(4)\ne 0,\Delta {{P}^{(9)}}(5)\ne 0 \\ 
\end{aligned} \right..
\end{equation}
From Eq.~\eqref{eq26} and Tab.~\ref{tab4}, we obtain
\begin{equation}
  \label{eq27}
  \left\{ \begin{aligned}
    & \left\{ C(1) \right\}\to \left\{ P(4),P(6),P(7) \right\} \\ 
   & \left\{ C(2) \right\}\to \left\{ P(1),P(2),P(6),P(7) \right\} \\ 
   & \left\{ C(3) \right\}\to \left\{ P(2),P(3),P(7),P(8) \right\} \\ 
   & \left\{ C(4) \right\}\to \left\{ P(3),P(4),P(5),P(9) \right\} \\ 
   & \left\{ C(5) \right\}\to \left\{ P(3),P(9) \right\} \\ 
   & \left\{ C(6) \right\}\to \left\{ P(3),P(5),P(6),P(8) \right\} \\ 
   & \left\{ C(7) \right\}\to \left\{ P(1),P(5),P(6),P(8) \right\} \\ 
   & \left\{ C(8) \right\}\to \left\{ P(1),P(2),P(8) \right\} \\ 
   & \left\{ C(9) \right\}\to \left\{ P(4),P(5) \right\} \\ 
  \end{aligned} \right.\xrightarrow{\text{Table 4}}\left\{ \begin{aligned}
    & \left\{ P(1) \right\}\leftarrow \left\{ C(2),C(7),C(8) \right\} \\ 
   & \left\{ P(2) \right\}\leftarrow \left\{ C(2),C(3),C(8) \right\} \\ 
   & \left\{ P(3) \right\}\leftarrow \left\{ C(3),C(4),C(5),C(6) \right\} \\ 
   & \left\{ P(4) \right\}\leftarrow \left\{ C(1),C(4),C(9) \right\} \\ 
   & \left\{ P(5) \right\}\leftarrow \left\{ C(4),C(6),C(7),C(9) \right\} \\ 
   & \left\{ P(6) \right\}\leftarrow \left\{ C(1),C(2),C(6),C(7) \right\} \\ 
   & \left\{ P(7) \right\}\leftarrow \left\{ C(1),C(2),C(3) \right\} \\ 
   & \left\{ P(8) \right\}\leftarrow \left\{ C(3),C(6),C(7),C(8) \right\} \\ 
   & \left\{ P(9) \right\}\leftarrow \left\{ C(4),C(5) \right\} \\ 
  \end{aligned} \right.  
\end{equation}
\begin{table}[!h]
  \centering
  \caption{Converting the mapping from ciphertext sets to plaintext sets in Eq.~\eqref{eq27} into a mapping from plaintext sets to ciphertext sets}
  \label{tab4}
\begin{tabular}{cccccccccc}
    \hline
         & $C(1)$ & $C(2)$ & $C(3)$ & $C(4)$ & $C(5)$ & $C(6)$ & $C(7)$ & $C(8)$ & $C(9)$ \\ \hline
    $P(1)$ &      & $\checkmark$    &      &      &      &      & $\checkmark$    & $\checkmark$    &      \\ 
    $P(2)$ &      & $\checkmark$    & $\checkmark$    &      &      &      &      & $\checkmark$    &      \\ 
    $P(3)$ &      &      & $\checkmark$    & $\checkmark$    & $\checkmark$    & $\checkmark$    &      &      &      \\ 
    $P(4)$ & $\checkmark$    &      &      & $\checkmark$    &      &      &      &      & $\checkmark$    \\ 
    $P(5)$ &      &      &      & $\checkmark$    &      & $\checkmark$    & $\checkmark$    &      & $\checkmark$    \\ 
    $P(6)$ & $\checkmark$    & $\checkmark$    &      &      &      & $\checkmark$    & $\checkmark$    &      &      \\ 
    $P(7)$ & $\checkmark$    & $\checkmark$    & $\checkmark$    &      &      &      &      &      &      \\ 
    $P(8)$ &      &      & $\checkmark$    &      &      & $\checkmark$    & $\checkmark$    & $\checkmark$    &      \\ 
    $P(9)$ &      &      &      & $\checkmark$    & $\checkmark$    &      &      &      &      \\ \hline
    \end{tabular}
  \end{table}

  Although Fig.~\ref{fig9} shows that the positional association itself forms a complete chain of length 9, Eq.~\eqref{eq27} comprises associated positions obtained only via chosen-ciphertext attack and is not yet a complete chain of length 9. The next step is to obtain such a chain by reordering according to the six linking rules above, thereby recovering the permutation. Applying the six rules to reorder Eq.~\eqref{eq27} yields the result in Fig.~\ref{fig11}, which clearly forms a complete chain of length 9.
  \begin{figure}[!h] 
    \centering 
    \includegraphics[width=0.65\textwidth]{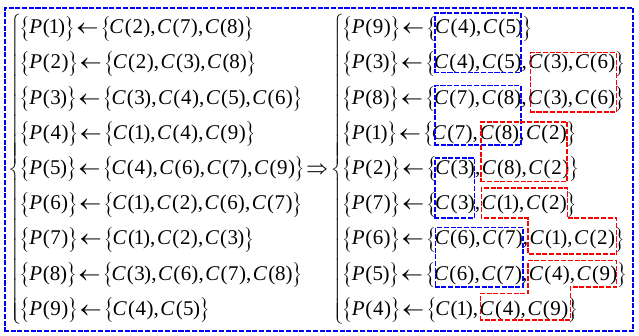} 
    \caption{A complete chain of length 9 obtained by reordering Eq.~\eqref{eq27} using the six linking rules.} 
    \label{fig11} 
  \end{figure}

  From Fig.~\ref{fig11}, the relation between plaintext positions and serial indices is shown in Fig.~\ref{fig12}. Reordering by plaintext serial index yields Fig.~\ref{fig13}. From Fig.~\ref{fig13}, it is evident that we have recovered $$f_{S0}^{-1}:(1,4),(2,5),(3,2),(4,9),(5,8),(6,7),(7,6),(8,3),(9,1).$$
  \begin{figure}[!h] 
    \centering 
    \includegraphics[width=0.58\textwidth]{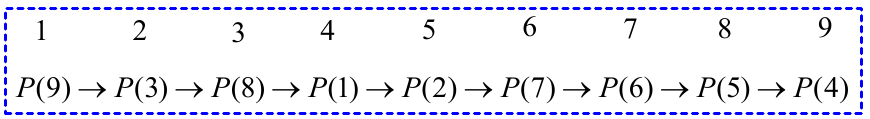} 
    \caption{Relationship between plaintext positions and serial indices.} 
    \label{fig12} 
  \end{figure}
  \begin{figure}[!h] 
    \centering 
    \includegraphics[width=0.58\textwidth]{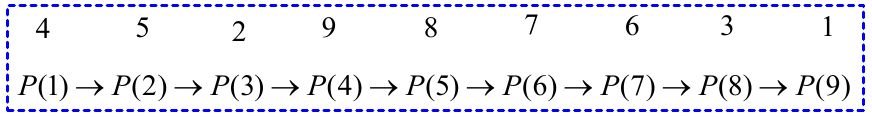} 
    \caption{Result after reordering by plaintext serial indices.} 
    \label{fig13}
  \end{figure}
  (3) 
  Basic method for further recovering the mid-permutation $f_{S1}^{-1}$ and the post-permutation $f_{S2}^{-1}$ via the chain attack

As analyzed with Fig.~\ref{fig10}, once $f_{S0}^{-1}$ is recovered, $P_{1}$ becomes known to the attacker. The attacker then recovers the subsequent four-stage PNDCC (diffusion-permutation-diffusion-permutation). While attacking the mid-permutation, one needs to reduce the number of ciphertext elements in the sets from 4 to 2 by taking unions over plaintext and intersections over ciphertext. The corresponding iterative equations and the result after applying plaintext unions and ciphertext intersections are
\begin{equation}
  \label{eq28}
\begin{aligned}
  & \left\{ \begin{aligned}
  & {{P}_{1}}(k)=f_{D1}^{-1}\left( \begin{aligned}
  & f_{D2}^{-1}\left( C\left( f_{S2}^{-1}\left( f_{S1}^{-1}\left( k \right) \right) \right), C\left( f_{S2}^{-1}\left( f_{S1}^{-1}\left( k \right)-1 \right) \right), {{K}_{2}}\left( f_{S1}^{-1}\left( k \right) \right)   
  \right), \\ 
  & f_{D2}^{-1}\left( C\left( f_{S2}^{-1}\left( f_{S1}^{-1}\left( k-1 \right) \right) \right), C\left( f_{S2}^{-1}\left( f_{S1}^{-1}\left( k-1 \right)-1 \right) \right), {{K}_{2}}\left( f_{S1}^{-1}\left( k-1 \right) \right)   
  \right),\\ 
  &{{K}_{1}}\left( k \right) \\ 
\end{aligned} \right) \\ 
 & {{P}_{1}}(k+1)=f_{D1}^{-1}\left( \begin{aligned}
  & f_{D2}^{-1}\left( C\left( f_{S2}^{-1}\left( f_{S1}^{-1}\left( k+1 \right) \right) \right), C\left( f_{S2}^{-1}\left( f_{S1}^{-1}\left( k+1 \right)-1 \right) \right), {{K}_{2}}\left( f_{S1}^{-1}\left( k+1 \right) \right)   
  \right), \\ 
  & f_{D2}^{-1}\left( C\left( f_{S2}^{-1}\left( f_{S1}^{-1}\left( k \right) \right) \right), C\left( f_{S2}^{-1}\left( f_{S1}^{-1}\left( k \right)-1 \right) \right), {{K}_{2}}\left( f_{S1}^{-1}\left( k \right) \right)   
  \right),\\ 
  &{{K}_{1}}\left( k+1 \right) \\ 
\end{aligned} \right) \\ 
\end{aligned} \right. \\ 
 & \to \left\{ \begin{aligned}
  & \left\{ {{P}_{1}}(k) \right\}\leftarrow \left\{ \begin{aligned}
  & C\left( f_{S2}^{-1}\left( f_{S1}^{-1}\left( k \right) \right) \right),C\left( f_{S2}^{-1}\left( f_{S1}^{-1}\left( k \right)-1 \right) \right), \\ 
 & C\left( f_{S2}^{-1}\left( f_{S1}^{-1}\left( k-1 \right) \right) \right),C\left( f_{S2}^{-1}\left( f_{S1}^{-1}\left( k-1 \right)-1 \right) \right) \\ 
\end{aligned} \right\} \\ 
 & \left\{ {{P}_{1}}(k+1) \right\}\leftarrow \left\{ \begin{aligned}
  & C\left( f_{S2}^{-1}\left( f_{S1}^{-1}\left( k+1 \right) \right) \right),C\left( f_{S2}^{-1}\left( f_{S1}^{-1}\left( k+1 \right)-1 \right) \right), \\ 
 & C\left( f_{S2}^{-1}\left( f_{S1}^{-1}\left( k \right) \right) \right),C\left( f_{S2}^{-1}\left( f_{S1}^{-1}\left( k \right)-1 \right) \right) \\ 
\end{aligned} \right\} \\ 
\end{aligned} \right. \\ 
 & \xrightarrow[\text{Ciphertext: intersection}]{\text{Plaintext: union}} \left\{ {{P}_{1}}(k),{{P}_{1}}(k+1) \right\}\leftarrow \left\{
 C\left( f_{S2}^{-1}\left( f_{S1}^{-1}\left( k \right) \right) \right), C\left( f_{S2}^{-1}\left( f_{S1}^{-1}\left( k \right)-1 \right) \right)   
\right\} \\ 
 & \to \left\{ {{P}_{1}}\left( f_{S1}^{{}}\left( k \right) \right),{{P}_{1}}\left( f_{S1}^{{}}\left( k \right)+1 \right) \right\}\leftarrow \left\{C\left( f_{S2}^{-1}\left( k \right) \right), C\left( f_{S2}^{-1}\left( k-1 \right) \right)   
\right\} \\ 
\end{aligned}.
\end{equation}

The known mid-permutation $f_{S1}^{-1},f_{S1}^{{}}$ and post-permutation $f_{S2}^{-1}$ in Eq.~\eqref{eq28} are
\begin{equation}
  \label{eq29}
  \left\{ \begin{aligned}
  & f_{S1}^{-1}:(1,9),(2,2),(3,4),(4,5),(5,1),(6,6),(7,3),(8,8),(9,7) \\ 
 & f_{S1}:(1,5),(2,2),(3,7),(4,3),(5,4),(6,6),(7,9),(8,8),(9,1) \\ 
 & f_{S2}^{-1}:(1,3),(2,6),(3,7),(4,8),(5,2),(6,1),(7,9),(8,4),(9,5) \\ 
\end{aligned} \right..
\end{equation}

From the decryption machine in Fig.~\ref{fig10}, the chosen-ciphertext attack, and Tab.~\ref{tab5}, we obtain
\begin{equation}
  \label{eq30}
\left\{ \begin{aligned}
  & \left\{ C(1) \right\}\to \left\{ {{P}_{1}}(6),{{P}_{1}}(7),{{P}_{1}}(9) \right\} \\ 
 & \left\{ C(2) \right\}\to \left\{ {{P}_{1}}(4),{{P}_{1}}(5),{{P}_{1}}(6),{{P}_{1}}(7) \right\} \\ 
 & \left\{ C(3) \right\}\to \left\{ {{P}_{1}}(2),{{P}_{1}}(3),{{P}_{1}}(5),{{P}_{1}}(6) \right\} \\ 
 & \left\{ C(4) \right\}\to \left\{ {{P}_{1}}(1),{{P}_{1}}(2),{{P}_{1}}(8),{{P}_{1}}(9) \right\} \\ 
 & \left\{ C(5) \right\}\to \left\{ {{P}_{1}}(1),{{P}_{1}}(2) \right\} \\ 
 & \left\{ C(6) \right\}\to \left\{ {{P}_{1}}(2),{{P}_{1}}(3),{{P}_{1}}(7),{{P}_{1}}(8) \right\} \\ 
 & \left\{ C(7) \right\}\to \left\{ {{P}_{1}}(3),{{P}_{1}}(4),{{P}_{1}}(7),{{P}_{1}}(8) \right\} \\ 
 & \left\{ C(8) \right\}\to \left\{ {{P}_{1}}(3),{{P}_{1}}(4),{{P}_{1}}(5) \right\} \\ 
 & \left\{ C(9) \right\}\to \left\{ {{P}_{1}}(8),{{P}_{1}}(9) \right\} \\ 
\end{aligned} \right.\xrightarrow{\text{Table 5}}\left\{ \begin{aligned}
  & \left\{ {{P}_{1}}(1) \right\}\leftarrow \left\{ C(4),C(5) \right\} \\ 
 & \left\{ {{P}_{1}}(2) \right\}\leftarrow \left\{ C(3),C(4),C(5),C(6) \right\} \\ 
 & \left\{ {{P}_{1}}(3) \right\}\leftarrow \left\{ C(3),C(6),C(7),C(8) \right\} \\ 
 & \left\{ {{P}_{1}}(4) \right\}\leftarrow \left\{ C(2),C(7),C(8) \right\} \\ 
 & \left\{ {{P}_{1}}(5) \right\}\leftarrow \left\{ C(2),C(3),C(8) \right\} \\ 
 & \left\{ {{P}_{1}}(6) \right\}\leftarrow \left\{ C(1),C(2),C(3) \right\} \\ 
 & \left\{ {{P}_{1}}(7) \right\}\leftarrow \left\{ C(1),C(2),C(6),C(7) \right\} \\ 
 & \left\{ {{P}_{1}}(8) \right\}\leftarrow \left\{ C(4),C(6),C(7),C(9) \right\} \\ 
 & \left\{ {{P}_{1}}(9) \right\}\leftarrow \left\{ C(1),C(4),C(9) \right\} \\ 
\end{aligned} \right.
\end{equation}

\begin{table}[!h]
  \centering
  \caption{Converting the mapping from ciphertext sets to plaintext sets in Eq.~\eqref{eq30} into a mapping from plaintext sets to ciphertext sets}
  \label{tab5}
\begin{tabular}{cccccccccc}
  \hline
        & $C(1)$ & $C(2)$ & $C(3)$ & $C(4)$ & $C(5)$ & $C(6)$ & $C(7)$ & $C(8)$ & $C(9) $\\ \hline
  $P1(1)$ &      &      &      & $\checkmark$    & $\checkmark$    &      &      &      &      \\ 
  $P1(2)$ &      &      & $\checkmark$    & $\checkmark$    & $\checkmark$    & $\checkmark$    &      &      &      \\ 
  $P1(3)$ &      &      & $\checkmark$    &      &      & $\checkmark$    & $\checkmark$    & $\checkmark$    &      \\ 
  $P1(4)$ &      & $\checkmark$    &      &      &      &      & $\checkmark$    & $\checkmark$    &      \\ 
  $P1(5)$ &      & $\checkmark$    & $\checkmark$    &      &      &      &      & $\checkmark$    &      \\ 
  $P1(6)$ & $\checkmark$    & $\checkmark$    & $\checkmark$    &      &      &      &      &      &      \\ 
  $P1(7)$ & $\checkmark$    & $\checkmark$    &      &      &      & $\checkmark$    & $\checkmark$    &      &      \\ 
  $P1(8)$ &      &      &      & $\checkmark$    &      & $\checkmark$    & $\checkmark$    &      & $\checkmark$    \\ 
  $P1(9)$ & $\checkmark$    &      &      & $\checkmark$    &      &      &      &      & $\checkmark$    \\ \hline
  \end{tabular}
  \end{table}

  According to the structure of Eq.~\eqref{eq28}, we need to take intersections over adjacent ciphertext sets and unions over adjacent plaintext sets in Eq.~\eqref{eq30}. Then, using the six linking rules, we obtain the result shown in Fig.~\ref{fig14}. Comparing with $k$, we see that both the mid-permutation $f_{S1}$ and the post-permutation $f_{S2}^{-1}$ in Eq.~\eqref{eq29} are simultaneously recovered.
  \begin{figure}[htbp] 
    \centering 
    \includegraphics[width=0.68\textwidth]{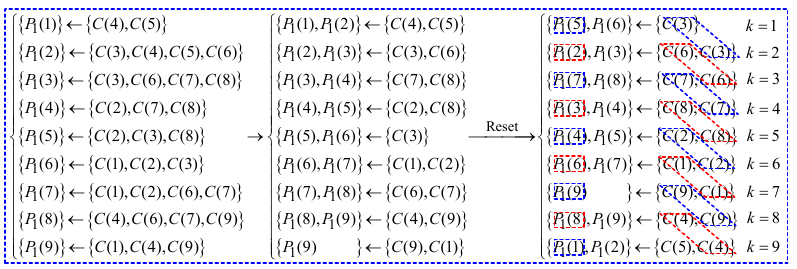} 
    \caption{Illustration of recovering the mid-permutation $f_{S1}$ and the post-permutation $f_{S2}^{-1}$.} 
    \label{fig14} 
  \end{figure}
  
  It should be emphasized that although this section presents a five-stage PNDCC with only 9 pixels and the permutation keys specified by Eq.~\eqref{eq21}, the basic chain attack method proposed here applies equally to multi-stage PNDCC with arbitrary $M\times N$ pixels and arbitrary permutation keys. In the general case, the chain attack flowchart is shown in Fig.~\ref{fig15}.

  \begin{figure}[htbp] 
    \centering 
    \includegraphics[width=0.57\textwidth]{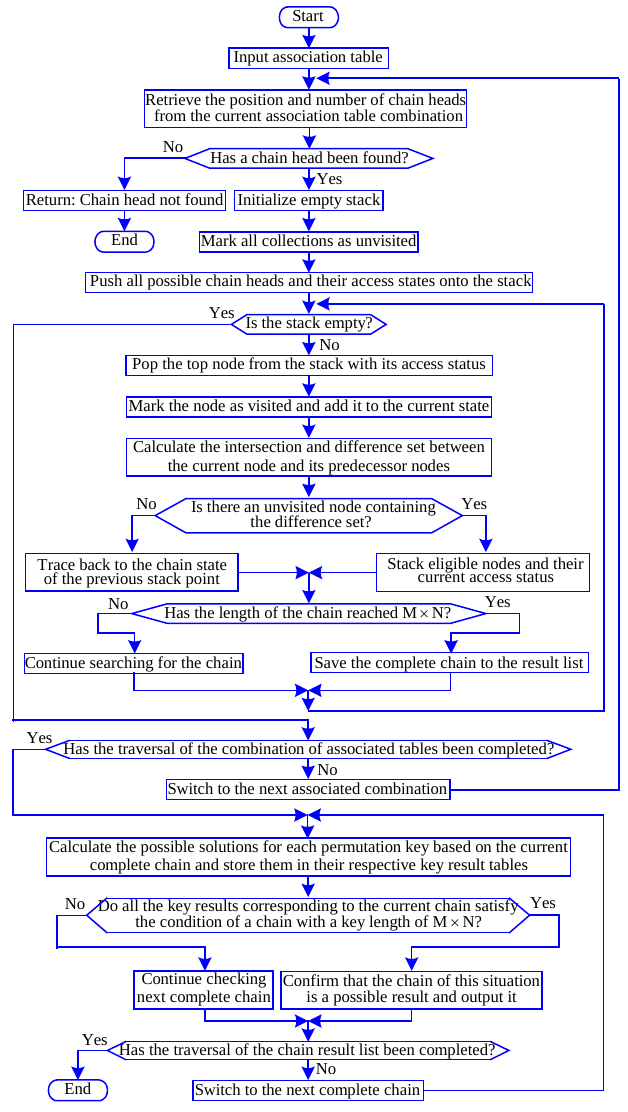}
    \caption{Flowchart of the chain attack in the general case.} 
    \label{fig15} 
  \end{figure}

\section{Recovering Diffusion Keys After All Permutations Are Broken}
\label{sec:diffusion-attack}

We use the five-stage PNDCC in Eq.~\eqref{eq14} as a representative example to explain how to recover the diffusion keys. The method generalizes to arbitrary multi-stage settings. Set ${{m}_{1}}={{m}_{2}}=1$; the resulting iteration is
\begin{equation}
  \label{eq31}
P\left( f_{S0}^{{}}\left( k \right) \right)=f_{D1}^{-1}\left( \begin{aligned}
  & f_{D2}^{-1}\left( C\left( f_{S2}^{-1}\left( f_{S1}^{-1}\left( k \right) \right) \right),C\left( f_{S2}^{-1}\left( f_{S1}^{-1}\left( k \right)-1 \right) \right),{{K}_{2}}\left( f_{S1}^{-1}\left( k \right) \right) 
 \right), \\ 
 & f_{D2}^{-1}\left( C\left( f_{S2}^{-1}\left( f_{S1}^{-1}\left( k-1 \right) \right) \right),C\left( f_{S2}^{-1}\left( f_{S1}^{-1}\left( k-1 \right)-1 \right) \right),{{K}_{2}}\left( f_{S1}^{-1}\left( k-1 \right) \right) 
\right), {{K}_{1}}\left( k \right) 
\end{aligned} \right),
\end{equation}
where $k=1,2,\cdots ,M\times N$ and $C(0),C(-1),{{K}_{1}}(0),{{K}_{2}}(0)$ are initial conditions; for convenience, all set to 0. We now show how to recover all diffusion keys once the permutation keys have been fully recovered.

Notice that in Eq.~\eqref{eq31} there are three unknowns to be solved at each position, namely ${{K}_{2}}(f_{S1}^{-1}(k))$, ${{K}_{2}}(f_{S1}^{-1}(k-1))$ and ${{K}_{1}}(k)$ for $k=1,2,\ldots,M\times N$. Under a chosen-ciphertext attack, choose $C^{(0)}=0$ to obtain the corresponding $P^{(0)}$. Choose $C^{(V_{1})}=V_{1}$ with $V_{1}\in\{1,2,\ldots,255\}$ to obtain the corresponding $P^{(V_{1})}$. Similarly, choose $C^{(V_{2})}=V_{2}$ with $V_{2}\in\{1,2,\ldots,255\}$ and $V_{2}\neq V_{1}$ to obtain the corresponding $P^{(V_{2})}$. From Eq.~\eqref{eq31}, we have
\begin{equation}
  \label{eq32}
  \left\{ \begin{aligned}
  & P^{(0)}\left( f_{S0}^{{}}\!\left( k \right) \right)\!=\!f_{D1}^{-1}\!\left( \begin{aligned}
  & f_{D2}^{-1}\!\left(C^{(0)}\!\left( f_{S2}^{-1}\!\left( f_{S1}^{-1}\!\left( k \right) \right) \right),C^{(0)}\!\left( f_{S2}^{-1}\!\left( f_{S1}^{-1}\!\left( k \right)\!-\!1 \right) \right),{{K}_{2}}\!\left( f_{S1}^{-1}\!\left( k \right) \right) 
\right), \\ 
 & f_{D2}^{-1}\!\left( C^{(0)}\!\left( f_{S2}^{-1}\!\left( f_{S1}^{-1}\!\left( k\!-\!1 \right) \right) \right),C^{(0)}\!\left( f_{S2}^{-1}\!\left( f_{S1}^{-1}\!\left( k\!-\!1 \right)\!-\!1 \right) \right),{{K}_{2}}\!\left( f_{S1}^{-1}\!\left( k\!-\!1 \right) \right)
 \right),{{K}_{1}}\!\left( k \right) \\ 
\end{aligned} \right) \\ 
 & P^{({{V}_{1}})}\!\left( f_{S0}^{{}}\!\left( k \right) \right)\!=\!f_{D1}^{-1}\!\left( \begin{aligned}
  & f_{D2}^{-1}\!\left( C^{({{V}_{1}})}\!\left( f_{S2}^{-1}\!\left( f_{S1}^{-1}\!\left( k \right) \right) \right),C^{({{V}_{1}})}\!\left( f_{S2}^{-1}\!\left( f_{S1}^{-1}\!\left( k \right)\!-\!1 \right) \right),{{K}_{2}}\!\left( f_{S1}^{-1}\!\left( k \right) \right) 
\right), \\ 
 & f_{D2}^{-1}\!\left( C^{({{V}_{1}})}\!\left( f_{S2}^{-1}\!\left( f_{S1}^{-1}\!\left( k\!-\!1 \right) \right) \right),C^{({{V}_{1}})}\!\left( f_{S2}^{-1}\!\left( f_{S1}^{-1}\!\left( k\!-\!1 \right)\!-\!1 \right) \right),{{K}_{2}}\!\left( f_{S1}^{-1}\!\left( k\!-\!1 \right) \right) 
 \right),{{K}_{1}}\!\left( k \right) \\ 
\end{aligned} \right) \\ 
 & P^{({{V}_{2}})}\!\left( f_{S0}^{{}}\!\left( k \right) \right)\!=\!f_{D1}^{-1}\!\left( \begin{aligned}
  & f_{D2}^{-1}\!\left( C^{({{V}_{2}})}\!\left( f_{S2}^{-1}\!\left( f_{S1}^{-1}\!\left( k \right) \right) \right),C^{({{V}_{2}})}\!\left( f_{S2}^{-1}\!\left( f_{S1}^{-1}\!\left( k \right)\!-\!1 \right) \right),{{K}_{2}}\!\left( f_{S1}^{-1}\!\left( k \right) \right)  
\right), \\ 
 & f_{D2}^{-1}\!\left( C^{({{V}_{2}})}\!\left( f_{S2}^{-1}\!\left( f_{S1}^{-1}\!\left( k\!-\!1 \right) \right) \right),C^{({{V}_{2}})}\!\left( f_{S2}^{-1}\!\left( f_{S1}^{-1}\!\left( k\!-\!1 \right)\!-\!1 \right) \right),{{K}_{2}}\!\left( f_{S1}^{-1}\!\left( k\!-\!1 \right) \right) 
\right),{{K}_{1}}\!\left( k \right) \\ 
\end{aligned} \right) \\ 
\end{aligned} \right.  .
\end{equation}
where $C^{(0)}=0$, $P^{(0)}$, $C^{(V_{1})}=V_{1}$ with $V_{1}\in\{1,2,\ldots,255\}$, $P^{(V_{1})}$, $C^{(V_{2})}=V_{2}$ with $V_{2}\in\{1,2,\ldots,255\}$ and $V_{2}\neq V_{1}$, and $P^{(V_{2})}$ are all known quantities. Since the permutation keys $f_{S0}$, $f_{S1}^{-1}$, and $f_{S2}^{-1}$ have already been recovered, all the following terms in the above system are known to the attacker:
\begin{equation}
  \label{eq33}
\left\{ \!\begin{aligned}
  & P^{(0)}\!\left(\! f_{S0}^{{}}\left( k \right) \!\right)\!,C^{(0)}\!\left(\! f_{S2}^{-\!1}\left(\! f_{S1}^{-\!1}\left( k \right) \!\right) \!\right)\!,C^{(0)}\!\left(\! f_{S2}^{-\!1}\left(\! f_{S1}^{-\!1}\left( k \right)\!-\!1 \right) \!\right)\!,C^{(0)}\!\left(\! f_{S2}^{-\!1}\left(\! f_{S1}^{-\!1}\left( k\!-\!1 \right) \!\right) \!\right)\!,C^{(0)}\!\left(\! f_{S2}^{-\!1}\left(\! f_{S1}^{-\!1}\left( k\!-\!1 \right)\!-\!1 \right) \!\right) \\ 
 & P^{({{V}_{1}})}\!\left(\! f_{S0}^{{}}\!\left(\! k \right) \!\right)\!,C^{({{V}_{1}})}\!\left(\! f_{S2}^{-\!1}\!\left(\! f_{S1}^{-\!1}\!\left( k \right) \!\right) \!\right)\!,C^{({{V}_{1}})}\!\left(\! f_{S2}^{-\!1}\!\left(\! f_{S1}^{-\!1}\!\left( k \right)\!-\!1 \right) \!\right)\!,C^{({{V}_{1}})}\!\left(\! f_{S2}^{-\!1}\!\left(\! f_{S1}^{-\!1}\!\left( k\!-\!1 \right) \!\right) \!\right)\!,C^{({{V}_{1}})}\!\left(\! f_{S2}^{-\!1}\!\left(\! f_{S1}^{-\!1}\!\left( k\!-\!1 \right)\!-\!1 \right) \!\right) \\ 
 & P^{({{V}_{2}})}\!\left(\! f_{S0}^{{}}\!\left( k \right) \!\right)\!,C^{({{V}_{2}})}\!\left(\! f_{S2}^{-\!1}\!\left(\! f_{S1}^{-\!1}\!\left( k \right) \!\right) \!\right)\!,C^{({{V}_{2}})}\!\left(\! f_{S2}^{-\!1}\!\left(\! f_{S1}^{-\!1}\!\left( k \right)\!-\!1 \right) \!\right)\!,C^{({{V}_{2}})}\!\left(\! f_{S2}^{-\!1}\!\left(\! f_{S1}^{-\!1}\!\left( k\!-\!1 \right) \!\right) \!\right)\!,C^{({{V}_{2}})}\!\left(\! f_{S2}^{-\!1}\!\left(\! f_{S1}^{-\!1}\!\left( k\!-\!1 \right)\!-\!1 \right) \!\right) \\ 
\end{aligned} \right.\!.
\end{equation}
Therefore, the diffusion keys can be solved from Eq.~\eqref{eq32}.

Setting $k=1$ in Eq.~\eqref{eq32} yields
\begin{equation}
  \label{eq34}
\left\{ \begin{aligned}
  & P^{(0)}\left( f_{S0}^{{}}\left( 1 \right) \right)=f_{D1}^{-1}\left( \begin{aligned}
  & f_{D2}^{-1}\left( C^{(0)}\left( f_{S2}^{-1}\left( f_{S1}^{-1}\left( 1 \right) \right) \right),C^{(0)}\left( f_{S2}^{-1}\left( f_{S1}^{-1}\left( 1 \right)-1 \right) \right),{{K}_{2}}\left( f_{S1}^{-1}\left( 1 \right) \right) 
 \right), \\ 
 & f_{D2}^{-1}\left(C^{(0)}\left( 0 \right),C^{(0)}\left( -1 \right),{{K}_{2}}\left( 0 \right) 
\right),{{K}_{1}}\left( 1 \right) \\ 
\end{aligned} \right) \\ 
 & P^{({{V}_{1}})}\left( f_{S0}^{{}}\left( 1 \right) \right)=f_{D1}^{-1}\left( \begin{aligned}
  & f_{D2}^{-1}\left( C^{({{V}_{1}})}\left( f_{S2}^{-1}\left( f_{S1}^{-1}\left( 1 \right) \right) \right),C^{({{V}_{1}})}\left( f_{S2}^{-1}\left( f_{S1}^{-1}\left( 1 \right)-1 \right) \right),{{K}_{2}}\left( f_{S1}^{-1}\left( 1 \right) \right) 
\right), \\ 
 & f_{D2}^{-1}\left( C^{({{V}_{1}})}\left( 0 \right),C^{({{V}_{1}})}\left( -1 \right),{{K}_{2}}\left( 0 \right) 
 \right),{{K}_{1}}\left( 1 \right) \\ 
\end{aligned} \right) \\ 
 & P^{({{V}_{2}})}\left( f_{S0}^{{}}\left( 1 \right) \right)=f_{D1}^{-1}\left( \begin{aligned}
  & f_{D2}^{-1}\left( C^{({{V}_{2}})}\left( f_{S2}^{-1}\left( f_{S1}^{-1}\left( 1 \right) \right) \right),C^{({{V}_{2}})}\left( f_{S2}^{-1}\left( f_{S1}^{-1}\left( 1 \right)-1 \right) \right),{{K}_{2}}\left( f_{S1}^{-1}\left( 1 \right) \right) 
\right), \\ 
 & f_{D2}^{-1}\left( C^{({{V}_{2}})}\left( 0 \right),C^{({{V}_{2}})}\left( -1 \right),{{K}_{2}}\left( 0 \right) 
 \right),{{K}_{1}}\left( 1 \right) \\ 
\end{aligned} \right) \\ 
\end{aligned} \right..
\end{equation}

Solving the system gives ${{K}_{2}}(f_{S1}^{-1}(1))$, $K_{2}(0)$, and $K_{1}(1)$. Setting $k=2$ in Eq.~\eqref{eq32} yields
\begin{equation}
  \label{eq35}
\left\{ \begin{aligned}
  & P^{(0)}\left( f_{S0}^{{}}\left( 2 \right) \right)=f_{D1}^{-1}\left( \begin{aligned}
  & f_{D2}^{-1}\left( C^{(0)}\left( f_{S2}^{-1}\left( f_{S1}^{-1}\left( 2 \right) \right) \right),C^{(0)}\left( f_{S2}^{-1}\left( f_{S1}^{-1}\left( 2 \right)-1 \right) \right),{{K}_{2}}\left( f_{S1}^{-1}\left( 2 \right) \right) 
\right), \\ 
 & f_{D2}^{-1}\left( C^{(0)}\left( f_{S2}^{-1}\left( f_{S1}^{-1}\left( 1 \right) \right) \right),C^{(0)}\left( f_{S2}^{-1}\left( f_{S1}^{-1}\left( 1 \right)-1 \right) \right),{{K}_{2}}\left( f_{S1}^{-1}\left( 1 \right) \right) 
\right),{{K}_{1}}\left( 2 \right) \\ 
\end{aligned} \right) \\ 
 & P^{({{V}_{1}})}\left( f_{S0}^{{}}\left( 2 \right) \right)=f_{D1}^{-1}\left( \begin{aligned}
  & f_{D2}^{-1}\left(  C^{({{V}_{1}})}\left( f_{S2}^{-1}\left( f_{S1}^{-1}\left( 2 \right) \right) \right),C^{({{V}_{1}})}\left( f_{S2}^{-1}\left( f_{S1}^{-1}\left( 2 \right)-1 \right) \right),{{K}_{2}}\left( f_{S1}^{-1}\left( 2 \right) \right) 
 \right), \\ 
 & f_{D2}^{-1}\left(C^{({{V}_{1}})}\left( f_{S2}^{-1}\left( f_{S1}^{-1}\left( 1 \right) \right) \right),C^{({{V}_{1}})}\left( f_{S2}^{-1}\left( f_{S1}^{-1}\left( 1 \right)-1 \right) \right),{{K}_{2}}\left( f_{S1}^{-1}\left( 1 \right) \right) 
\right),{{K}_{1}}\left( 2 \right) \\ 
\end{aligned} \right) \\ 
 & P^{({{V}_{2}})}\left( f_{S0}^{{}}\left( 2 \right) \right)=f_{D1}^{-1}\left( \begin{aligned}
  & f_{D2}^{-1}\left(  C^{({{V}_{2}})}\left( f_{S2}^{-1}\left( f_{S1}^{-1}\left( 2 \right) \right) \right),C^{({{V}_{2}})}\left( f_{S2}^{-1}\left( f_{S1}^{-1}\left( 2 \right)-1 \right) \right),{{K}_{2}}\left( f_{S1}^{-1}\left( 2 \right) \right) 
\right), \\ 
 & f_{D2}^{-1}\left( C^{({{V}_{2}})}\left( f_{S2}^{-1}\left( f_{S1}^{-1}\left( 1 \right) \right) \right),C^{({{V}_{2}})}\left( f_{S2}^{-1}\left( f_{S1}^{-1}\left( 1 \right)-1 \right) \right),{{K}_{2}}\left( f_{S1}^{-1}\left( 1 \right) \right)\right),{{K}_{1}}\left( 2 \right) \\ 
\end{aligned} \right) \\ 
\end{aligned} \right..
\end{equation}
Solving the system gives ${{K}_{2}}(f_{S1}^{-1}(2))$, ${{K}_{2}}(f_{S1}^{-1}(1))$, and ${{K}_{1}}(2)$.

Proceeding similarly for $k=3,4,\ldots,M\times N$, we obtain ${{K}_{2}}(f_{S1}^{-1}(k))$, ${{K}_{2}}(f_{S1}^{-1}(k-1))$, and $K_{1}(k)$ $(k=3,\allowbreak 4,\cdots ,M\times N)$. Because the diffusion is one-way and there is no mutual coupling among keys, one can independently solve for ${{K}_{2}}(f_{S1}^{-1}(1))$, ${{K}_{2}}(0)$, ${{K}_{1}}(1)$, ${{K}_{2}}(f_{S1}^{-1}(2))$, ${{K}_{2}}(f_{S1}^{-1}(1))$, ${{K}_{1}}(2)$, ${{K}_{2}}(f_{S1}^{-1}(k))$, ${{K}_{2}}(f_{S1}^{-1}(k-1))$, ${{K}_{1}}(k)$ $(k=3,4,\cdots ,M\times N)$, and so on. Moreover, the complexity of recovering the permutations via the chain attack is $M \times N + 1$ plaintext–ciphertext pairs. Since the diffusion keys can be directly recovered by reusing these pairs, no additional plaintext–ciphertext pairs are required. Therefore, the overall attack complexity for breaking a  multi-stage PNDCC remains $M \times N + 1$ plaintext–ciphertext pairs.

\section{An Example of the Chain Attack and an Improved Scheme PDCC for PNDCC}
\label{sec:attack-example-pdcc}
\subsection{Simulation experiment for breaking a four\allowbreak-stage diffusion\allowbreak-permutation\allowbreak-diffusion\allowbreak-permutation PNDCC}


According to the general chain attack flowchart shown in Fig.~\ref{fig15}, we take four-stage diffusion\allowbreak-permutation\allowbreak-diffusion\allowbreak-permutation  PNDCC as a typical example. 
The pre-diffusion and post-diffusion equations, the pre-diffusion and post-diffusion keys, the pre-permutation and post-permutation keys, and the simulation results are provided in supplementary material.


\subsection{An improved  PDCC scheme for PNDCC}

Based on Fig.~\ref{fig7} and Eqs.~\eqref{eq9}-\eqref{eq10}, we introduce plaintext delay terms to obtain
\begin{equation}
  \label{eq36}
\left\{ \begin{aligned}
  & {{P}_{1}}\left( k \right)\!=\!f_{D1}^{-1}\left( {{P}_{1}}\left( k-1 \right)\!,{{P}_{1}}\left( k-2 \right)\!,\cdots \!,{{P}_{1}}\left( k-{{n}_{1}} \right)\!,{{C}_{1}}\left( k \right)\!,{{C}_{1}}\left( k-1 \right)\!,{{C}_{1}}\left( k-2 \right)\!,\cdots \!,{{C}_{1}}\left( k-{{m}_{1}} \right)\!,{{K}_{1}}(k) 
 \right) \\ 
 & {{P}_{2}}\left( k \right)\!=\!f_{D2}^{-1}\left( {{P}_{2}}\left( k-1 \right)\!,{{P}_{2}}\left( k-2 \right)\!,\cdots \!,{{P}_{2}}\left( k-{{n}_{2}} \right)\!,{{C}_{2}}\left( k \right)\!,{{C}_{2}}\left( k-1 \right)\!,{{C}_{2}}\left( k-2 \right)\!,\cdots \!,{{C}_{2}}\left( k-{{m}_{2}} \right)\!,{{K}_{2}}(k) 
 \right) \\ 
\end{aligned} \right..
\end{equation}
where $n_{1},n_{2},m_{1},m_{2}$ are positive integers, $k=1,2,\ldots,M\times N$, and $f_{D1}^{-1}$ and $f_{D2}^{-1}$ are heterogeneous composite iterative functions~\cite{ref30}. We call a chaotic cipher whose diffusion equations in Fig.~\ref{fig7} satisfy Eq.~\eqref{eq36} a PDCC. Because PDCC includes plaintext delay terms, one cannot derive a form similar to Eq.~\eqref{eq14} from Fig.~\ref{fig7} and Eq.~\eqref{eq36}. As a result, the chain attack and codebook attacks are no longer applicable, and the scheme can resist attacks based on various cryptanalytic strategies.

\section{Conclusion}
\label{sec:conclusion}

This work proposes two novel cryptanalysis methods, ISBDA and the chain attack, successfully breaks a class of PNDCC, and presents an improved PDCC scheme that can resist algorithm-based cryptanalytic attacks. These results provide valuable guidance for the analysis and design of chaotic ciphers.

\section*{Conflicts of Interest}
The authors report no conflicts of interest. The authors alone are responsible for the content and writing of this paper.

\end{document}